\pdfoutput=1
\documentclass[letterpaper]{article}
\usepackage{aaai18}
\usepackage{times}
\usepackage{helvet}
\usepackage{courier}
\usepackage{url}
\usepackage{graphicx}
\usepackage{times}
\usepackage{amsmath}
\usepackage{amsthm}
\usepackage{mathtools}
\usepackage{amsfonts}
\usepackage{amssymb}
\usepackage{graphicx}

\usepackage{tikz}
\usetikzlibrary{arrows}
\usepackage{algorithm}
\usepackage{rotating}
\usepackage{pdflscape}
\usepackage[noend]{algpseudocode}

\usepackage{enumitem}
\usepackage{subfigure}
\usepackage{booktabs}
\usepackage[group-separator={,}]{siunitx}

\makeatletter
\newtheorem*{rep@theorem}{\rep@title}
\newcommand{\newreptheorem}[2]{%
\newenvironment{rep#1}[1]{%
 \def\rep@title{#2 \ref{##1}}%
 \begin{rep@theorem}}%
 {\end{rep@theorem}}}
\makeatother

\newreptheorem{theorem}{Theorem}

\DeclareMathOperator*{\argmax}{arg\,max}
\DeclareMathOperator{\sgn}{sgn}

\DeclarePairedDelimiter\floor{\lfloor}{\rfloor}

\newcommand{\englishcite}[1]{\citeauthor{#1}~\shortcite{#1}}
\newcommand{\POF}[2]{\textsf{POF}(#1, #2)}

\newcommand{\Secref}[1]{Section~\ref{#1}} 
 
\newcommand{\Appref}[1]{Appendix~\ref{#1}} 
 
\newcommand{\Eqref}[1]{Equation~(\ref{#1})} 

\newtheorem{lemma}{Lemma}
\newtheorem{theorem}{Theorem}
\newtheorem{proposition}{Proposition}

\newcommand{\AB}{\text{A-B}}
\newcommand{\BA}{\text{B-A}}
\newcommand{\OA}{\text{O-A}}
\newcommand{\OB}{\text{O-B}}
\newcommand{\BO}{\text{B-O}}
\newcommand{\AO}{\text{A-O}}

\newcommand{\OX}{\text{O-X}}
\newcommand{\XAB}{\text{X-AB}}

\newcommand{\ABO}{\text{AB-O}}
\newcommand{\OAB}{\text{O-AB}}
\newcommand{\AAB}{\text{A-AB}}
\newcommand{\ABA}{\text{AB-A}}
\newcommand{\BAB}{\text{B-AB}}
\newcommand{\ABB}{\text{AB-B}}
\renewcommand{\O}{\text{O}} 
\newcommand{\A}{\text{A}}
\newcommand{\B}{\text{B}}
\newcommand{\X}{\text{X}}
\newcommand{\XX}{\text{$X$-$X$}}
\newcommand{\XY}{\text{$X$-$Y$}}

\newcommand{\Pcal}{\mathcal{P}}
\newcommand{\Mcal}{\mathcal{M}}
\newcommand{\pbar}{\bar{p}}
\newcommand{\mAB}{\mu_\text{AB}}

\title{Balancing Lexicographic Fairness and a Utilitarian Objective\\with Application to Kidney Exchange}
\author{Duncan C. McElfresh$^{\dagger,\ddagger}$\\
$^\dagger$Department of Mathematics\\
University of Maryland\\
\texttt{dmcelfre@math.umd.edu}
\And
John P. Dickerson$^{\dagger,\ddagger}$\\
$^\ddagger$Department of Computer Science\\
University of Maryland\\
\texttt{john@cs.umd.edu}
}

\begin{document}

\maketitle

\begin{abstract}
Balancing fairness and efficiency in resource allocation is a classical economic and computational problem.  The price of fairness measures the worst-case loss of economic efficiency when using an inefficient but fair allocation rule; for indivisible goods in many settings, this price is unacceptably high.  One such setting is kidney exchange, where needy patients swap willing but incompatible kidney donors.  In this work, we close an open problem regarding the theoretical price of fairness in modern kidney exchanges.  We then propose a general hybrid fairness rule that balances a strict lexicographic preference ordering over classes of agents, and a utilitarian objective that maximizes economic efficiency. We develop a utility function for this rule that favors disadvantaged groups lexicographically; but if cost to overall efficiency becomes too high, it switches to a utilitarian objective. This rule has only one parameter which is proportional to a bound on the price of fairness, and can be adjusted by policymakers. We apply this rule to real data from a large kidney exchange and show that our hybrid rule produces more reliable outcomes than other fairness rules.
\end{abstract}

\section{Introduction}\label{sec:intro}
Chronic kidney disease is a worldwide problem whose societal burden is likened to that of diabetes~\cite{Neuen13:Global}.  Left untreated, it leads to end-stage renal failure and the need for a donor kidney---for which demand far outstrips supply.  In the United States alone, the kidney transplant waiting list grew from \num{58000} people in \num{2004} to over \num{100000} needy patients~\cite{Hart16:Kidney}.\footnote{{\tiny\texttt{https://optn.transplant.hrsa.gov/converge/data/}}}  

To alleviate some of this supply-demand mismatch, \emph{kidney exchanges}~\cite{Rapaport86:Case,Roth04:Kidney} allow patients with willing \emph{living} donors to trade donors for access to compatible or higher-quality organs. In addition to these patient-donor pairs, modern exchanges include \emph{non-directed donors}, who enter the exchange without a patient in need of a kidney. Exchanges occur in cycle- or chain-like structures, and now account for 10\% of living transplants in the United States.  Yet, access to a kidney exchange does not guarantee equal access to kidneys themselves; for example, certain classes of patients may be particularly disadvantaged based on health characteristics or other logistical factors.  Thus, \emph{fairness} considerations are an active topic of theoretical and practical research in kidney exchange and the matching market community in general.

Intuitively, any enforcement of a fairness constraint or consideration may have a negative effect on overall economic efficiency.  A quantification of this tradeoff is known as the \emph{price of fairness}~\cite{Bertsimas11:Price}.  Recent work by \englishcite{Dickerson14:Price} adapted this concept to the kidney exchange case, and presented two fair allocation rules that strike a balance between fairness and efficiency.  Yet, as we show in this paper, those rules can ``fail'' unpredictably, yielding an arbitrarily high price of fairness.

With this as motivation, we adapt to the kidney exchange case a recent technique for trading off a form of fairness and utilitarianism in a principled manner.  This technique is parameterized by a bound on the price of fairness, as opposed to a set of parameters that may result in hard-to-predict final matching behavior, as in past work.  We implement our rule in a realistic mathematical programming framework and--on real data from a large, multi-center, fielded kidney exchange--show that our rule effectively balances fairness and efficiency without unwanted outlier behavior.

\subsection{Related Work}\label{sec:intro-rw}
We briefly overview related work in balancing efficiency and fairness in resource allocation problem.  \englishcite{Bertsimas11:Price} define the price of fairness; that is, the relative loss in system efficiency under a fair allocation rule.  \englishcite{Hooker12:Combining} give a formal method for combining utilitarianism and equity.  We direct the reader to those two papers for a greater overview of research in fairness in general resource allocation problems.

Fairness in the context of kidney exchange was first studied by \englishcite{Roth05:Pairwise}; they explore concepts like Lorenz dominance in a stylized  model, and show that preferring fair allocations can come at great cost.  \englishcite{Li14:Egalitarian} extend this model and present an algorithm to solve for a Lorenz dominant matching.  Stability in kidney exchange, a concept intimately related to fairness, was explored by \englishcite{Liu14:Internally}.  The use of randomized allocation machanisms to promote fairness in stylized models is theoretically promising~\cite{Fang15:Randomized,Aziz16:Egalitarianism,Mattei17:Mechanisms}.  Recent work discusses fairness in stylized random graph models of dynamic kidney exchange~\cite{Ashlagi13:Kidney,Anderson15:Dynamic}.  None of these papers provide practical models that could be implemented in a fully-realistic and fielded kidney exchange.

Practically speaking, \englishcite{Yilmaz11:Kidney} explores in simulation equity issues from combining living and deceased donor allocation; that paper is limited to only short length-two kidney swaps, while real exchanges all use longer cycles and chains.  \englishcite{Dickerson14:Price} introduced two fairness rules explicitly in the context of kidney exchange, and proved bounds on the price of fairness under those rules in a random graph model; we build on that work in this paper, and describe it in greater detail later.  That work has been incorporated into a framework for learning to balance efficiency, fairness, and dynamism in matching markets~\cite{Dickerson15:FutureMatch}; we note that the fairness rule we present in this paper could be used in that framework as well.

\subsection{Our Contributions}\label{sec:intro-contributions}

\englishcite{Dickerson14:Price} finds that the theoretical price of fairness in kidney exchange is small when \emph{only} patient-donor pairs participate in the exchange.  They did not include non-directed donors (NDDs).  However, in modern kidney exchanges, non-directed donors (NDDs) provide many more matches than patient-donor pairs; furthermore, NDDs create more opportunities to expand the fair matching, potentially increasing the price of fairness.  Here, we prove that adding NDDs to the theoretical model actually \emph{decreases} the price of fairness, and that---with enough NDDs---the price of fairness is zero.

Real kidney exchanges are less dense and more uncertain than the (standard) theoretical model in which we prove our results. Previous approaches to incorporating fairness into kidney exchange have neglected this fact: they have been either ad-hoc---e.g., ``priority points'' decided on by committee~\cite{KPDPP13:CMR}---or brittle~\cite{Roth05:Pairwise,Dickerson14:Price}, resulting in an unacceptably high price of fairness. This paper provides the first approach to incorporating fairness into kidney exchange in a way that both prioritizes disadvantaged participants, but also comes with acceptable worst-case guarantees on the price of fairness.  Our method is easily applied as an objective in the mathematical-programming-based clearing methods used in today's fielded exchanges; indeed, using real data we show that this method guarantees a limit on efficiency loss.

\Secref{sec:prelims} introduces the kidney exchange problem. \Secref{sec:randomgraph} extends work by~\englishcite{Ashlagi14:Free} and~\englishcite{Dickerson14:Price}, showing that the price of fairness is small on the canonical random graph model even with NDDs. \Secref{sec:worstcase} shows that two recent fair allocation rules from the kidney exchange literature~\cite{Dickerson14:Price} can perform unacceptably poorly in the worst case. Then, \Secref{sec:combined} presents a new allocation rule that allows policymakers to set a limit on efficiency loss, while also favoring disadvantaged patients. \Secref{sec:experiments} shows on real data from a large fielded kidney exchange that our method limits efficiency loss while still favoring disadvantaged patients when possible.

\subsection{Preliminaries}\label{sec:prelims}
A kidney exchange can be represented as a directed \emph{compatibility graph} $G = (V,E)$, with vertices $V=P \cup N$ including both patient donor pairs $p \in P$ and non-directed-donors $n \in N$~\cite{Roth04:Kidney,Roth05:Kidney,Roth05:Pairwise,Abraham07:Clearing}. A directed edge $e$ is drawn from vertex $v_i$ to $v_j$ if the donor at $v_i$ can give to the patient at $v_j$.  Fielded kidney exchanges consist mainly of directed cycles in $G$, where each patient vertex in the cycle receives the donor kidney of the previous vertex. Modern exchanges also include non-cyclic structures called chains~\cite{Montgomery06:Domino,Rees09:Nonsimultaneous}.  Here, an NDD donates her kidney to a patient, whose paired donor donates her kidney to another patient, and so on.

In practice, cycles are limited in size, or ``capped,'' to some small constant $L$, while chains are limited in size to a much larger constant $R$---or not limited at all.  This is because all transplants in a cycle must execute \emph{simultaneously}; if a donor whose paired patient had already received a kidney backed out of the donation, then some participant in the market would be strictly worse off than before.  However, chains need not be executed simultaneously; if a donor backs out after her paired patient receives a kidney, then the chain breaks but no participant is strictly worse off.  We will discuss how these caps affect fairness and efficiency in the coming sections.

The goal of kidney exchange programs is to find a \emph{matching} $M$---a collection of disjoint cycles and chains in the graph $G$.  The cycles and chains must be disjoint because no donor can give more than one of her kidneys (although ongoing work explores multi-donor kidney exchange~\cite{Ergin17:Multi,Farina17:Operation}).  The \emph{clearing problem} in kidney exchange is to find a matching $M^*$ that maximizes some utility function $u : {\cal M} \to \mathbb{R}$, where ${\cal M}$ is the  set of all legal matchings.   Real kidney exchanges typically optimize for the maximum weighted cycle cover (i.e., $u(M) = \sum_{c \in M} \sum_{e \in c} w_{e}$).  This \emph{utilitarian} objective can favor certain classes of patient-donor pairs while disadvantaging others.  This is formalized in the following section.

\subsection{The Price of Fairness} 

As an example for this paper, we focus on \emph{highly-sensitized} patients, who have a very low probability of their blood passing a feasibility test with a random donor organ; thus, finding a kidney is often quite hard, and their median waiting time for an organ jumps by a factor of three over less sensitized patients.\footnote{\texttt{https://optn.transplant.hrsa.gov/data/}} Utilitarian objectives will, in general, marginalize these patients. Sensitization is determined using the Calculated Panel Reactive Antibody (CPRA) level of each patient, which reflects the likelihood that a patient will find a matching donor.

Formally the sensitization of each patient-donor vertex $v$ be $v_s \in [0,100]$, the CPRA level of $v$'s patient; NDD vertices are not associated with patients, so they do not have sensitization levels. Each patient-donor vertex $v \in P$ is considered highly sensitized if $v_s$ exceeds threshold $\tau\in [0,100]$, and lowly-sensitized otherwise. These vertex sets $V_H$ and $V_L$ are defined as: 
\begin{itemize}
  \item Lowly sensitized: $V_L = \{ v\ |\ v \in P : v_s < \tau \}$
  \item Highly sensitized: $V_H=\{ v\ |\ v \in P : v_s \geq \tau \}$.
\end{itemize}
By definition, highly-sensitized patients are harder to match than lowly-sensitized patients. Naturally, efficient matching algorithms prioritize easy-to-match vertices in $V_L$, marginalizing $V_H$. 
Let $u_f : \Mcal{} \to \mathbb{R}$ be a \emph{fair} utility function.  Formally, a utility function is fair when its corresponding optimal match $M^*_f$ is viewed as fair, where $M^*_f$ is defined as:
\begin{align*}
  M^*_f &= \argmax_{M \in \Mcal{}} u_f(M) 
\end{align*}
\englishcite{Bertsimas11:Price} defined the \emph{price of fairness} to be the ``relative system efficiency loss under a fair allocation assuming that a fully efficient allocation is one that maximizes the sum of [participant] utilities.''  \englishcite{Caragiannis09:Efficiency} defined an essentially identical concept in parallel.  Formally, given a fair utility function $u_f$ and the utilitarian utility function $u$, the price of fairness is:
\begin{align} \label{eq:pof}
  \POF{\Mcal{}}{u_f} = \frac{u\left(M^*_{\text{ }}\right) - u\left(M^*_f\right)}{u\left(M^*\right)}
\end{align}

The price of fairness $\POF{\Mcal{}}{u_f}$ is the relative loss in (utilitarian) efficiency caused by choosing the fair outcome $M^*_f$ rather than the most efficient outcome.

In the next section we show that the theoretical price of fairness in kidney exchange is small, even when both cycles \emph{and chains} are used---thus generalizing an earlier result due to~\englishcite{Dickerson14:Price} to modern kidney exchanges.

\section{%
The Theoretical Price of Fairness with Chains is Low (or Zero)%
}

\label{sec:randomgraph}
In this section we use the random graph model for kidney exchange introduced by~\englishcite{Ashlagi14:Free} to show that the theoretical price of fairness is always small, especially when NDDs are included. A complete description of this model can be found in \Appref{sec:randgraphmodel}. \englishcite{Dickerson14:Price} finds that without NDDs, the maximum price of fairness is $2/33$. Adding NDDs to this model creates more opportunities to match highly sensitized patients, which could potentially lead to a higher price of fairness. However we find that including chains in this model only \emph{decreases} the price of fairness; furthermore, when the ratio of NDDs to patient-donor pairs is high enough, the price of fairness is zero.

\subsection{Price of Fairness}

\englishcite{Ashlagi14:Free} characterize efficient matchings in a random graph model without chains, and \englishcite{Dickerson14:Price} build on this to show that the price of fairness without chains is bounded above by $2/33$. \englishcite{Dickerson12:Optimizing} extend the efficient matching of~\englishcite{Ashlagi14:Free} to include chains, but do not calculate the price of fairness.  In this work, we close the remaining gap in theory regarding the price of fairness with chains.

Given $|P|$ patient-donor pairs, we parameterize the number of NDDs $|N|$ with $\beta \geq 0$ such that $|N| = \beta|P|$.  Theorems \ref{thm:pofdecreases} and \ref{thm:maxbeta} state our two main results: adding chains to the random graph model does not increase the price of fairness, and when the fraction of NDDs is high enough ($\beta>1/8$), the price of fairness is zero. The proofs of the following theorems are given in \Appref{sec:pofrandomgraph}.

\begin{theorem}\label{thm:pofdecreases}
Adding NDDs to the random graph model ($\beta>0$) does not increase the upper bound on the price of fairness found by~\englishcite{Dickerson14:Price}.
\end{theorem}

\textbf{Proof Sketch: } We explore every possible efficient matching on the random graph model with chains; only four of these matchings have nonzero price of fairness. For each case, we compare the price of fairness to that of the efficient matching without chains found in~\englishcite{Dickerson14:Price}, and find that the upper bound does not increase.

\begin{theorem}\label{thm:maxbeta}
The price of fairness is zero when $\beta>1/8$.
\end{theorem}

\textbf{Proof sketch: } For each matching with nonzero price of fairness, $\beta \leq 1/8$. When $\beta>1/8$, a different matching occurs, and the price of fairness is zero.

\begin{figure}[h!]\label{fig:allcases-max}
\centering
\includegraphics[width=0.45\textwidth]{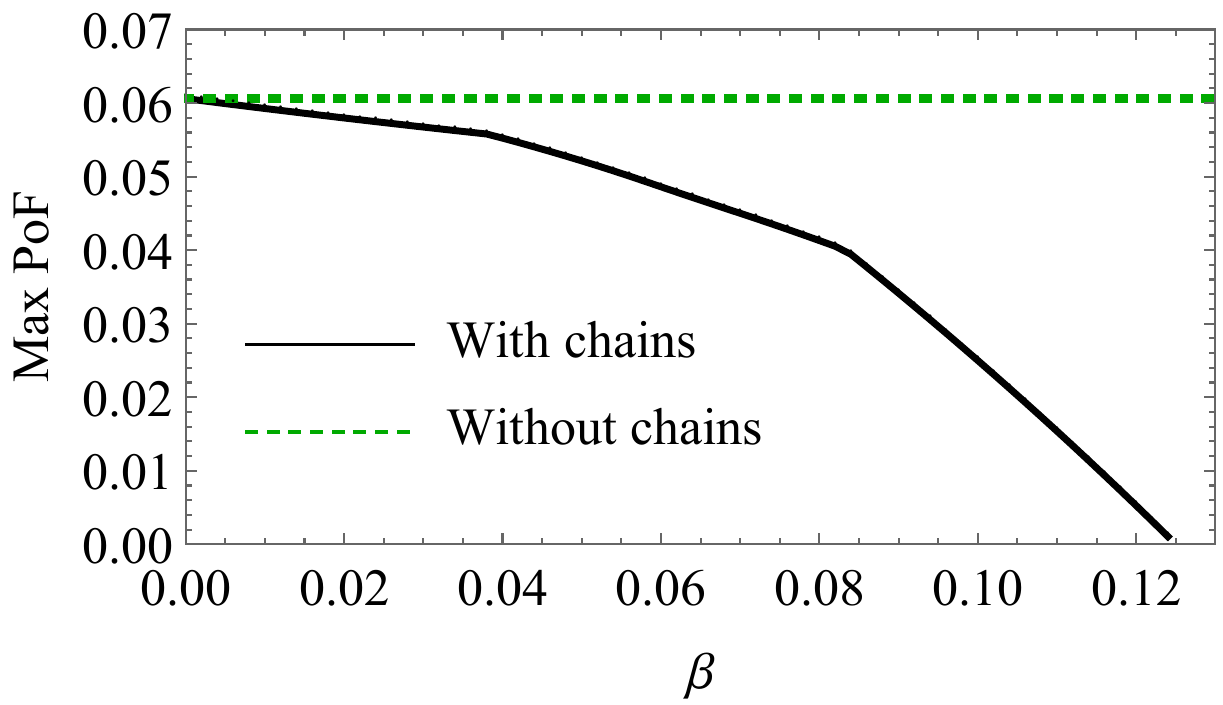}
\caption{Price of fairness with chains. (The horizontal dotted line at $2/33$ is the price of fairness without chains.)}
\end{figure}

To illustrate these results, we compute the price of fairness when $\beta \in [0,1/8]$. These calculations confirm our theoretical results, as shown in Figure~\ref{fig:allcases-max}: the price of fairness decreases as $\beta$ increases, and is zero when $\beta>1/8$.

The worst-case price of fairness is small in the random graph model, with or without NDDs. However, real exchange graphs are typically much sparser and less uniform---in reality the price of fairness can be high. In the next section, we discuss two notions of fairness in kidney exchange and determine their worst-case price of fairness.

\section{The Price of Fairness in State-of-the-Art Fair Rules can be Arbitrarily Bad}\label{sec:worstcase}

The price of fairness depends on how fairness is defined. This is especially true in real exchanges where the price of fairness can be unacceptably high.

In this section, we discuss two kidney-exchange-specific fairness rules introduced by~\englishcite{Dickerson14:Price}: lexicographic fairness and weighted fairness. These rules favor the disadvantaged class, or classes, without considering overall loss in efficiency; we will show in the worst case these rules allow the the price of fairness to approach $1$ (i.e., total efficiency loss). Proofs of these theorems are in \Appref{sec:poffairness}.

\subsection{Lexicographic Fairness}\label{ssec:lex}


As proposed by \englishcite{Dickerson14:Price}, $\alpha$-lexicographic fairness assigns nonzero utility only to matchings that award at least a fraction $\alpha$ of the maximum possible fair utility. Letting $u_H(M)$ and $u_L(M)$ be the utility assigned to only vertices in $V_H$ and $V_L$, respectively, the utility function for $\alpha$-lexicographic fairness is given in \Eqref{eq:alex}.
{\small
\begin{equation}
\label{eq:alex}
u_{\alpha}(M)=
\begin{cases}
u_L(M) + u_H(M) \\ \hspace{30pt} \text{ if } u_H(M)  \geq \alpha \max\limits_{M'\in \Mcal} u_H(M') \\
\vspace{5pt}
0  \hspace{25pt} \text{otherwise.}
\end{cases}
\end{equation}
}
Theorems~\ref{thm:lex-cycle} and~\ref{thm:lex-chain} state that strict lexicographic fairness ($\alpha=1$) allows the price of fairness to approach 1.
\begin{theorem}
\label{thm:lex-cycle}
For any cycle cap $L$ there exists a graph $G$ such that the price of fairness of $G$ under $\alpha$-lexicographic fairness with $0<\alpha\leq 1$ is bounded by {\small $\POF{\mathcal{M}}{u_{\alpha}} \geq \frac{L-2}{L}$}.
\end{theorem}
\begin{theorem}
\label{thm:lex-chain}
For any chain cap $R$ there exists a graph $G$ such that the price of fairness of $G$ under the $\alpha$-lexicographic fairness rule with $0<\alpha\leq 1$ is bounded by {\small $\POF{\mathcal{M}}{u_{\alpha}} \geq \frac{R-1}{R}$}.
\end{theorem}
Thus, $\alpha$-lexicographic fairness allows for a price of fairness that approaches $1$ as the cycle and chain cap increase. 

\subsection{Weighted Fairness}\label{ssec:weighted}

The weighted fairness rule~\cite{Dickerson14:Price} defines a utility function by first modifying the original edge weights $w_e$ by a multiplicative factor $\gamma \in \mathbb{R}$ such that
{\small
$$w_e' =
\begin{cases}
(1+\gamma)w_e &\text{if $e$ ends in $V_H$}  \\
w_e &\text{otherwise.}
\end{cases}$$}
Then the weighted fairness rule $u_{\mathit{WF}}$ is
$$ u_{WF}(M) = \sum_{c \in M} u'(c), $$
where $u'(c)$ is the utility of a chain or cycle $c$ with modified edge weights.

The modified edge weights prompt the matching algorithm to include more highly-sensitized patients; as in the lexicographic case, we now show that the price of fairness approaches $1$ under weighted fairness.

\begin{theorem}\label{thm:weight-cycle}
For any cycle cap $L$ and $\gamma \geq L-1$, there exists a graph $G$ such that the price of fairness of $G$ under the weighted fairness rule is bounded by {\small $\POF{\mathcal{M}}{u_{WF}} \geq \frac{L-2}{L}$}.
\end{theorem}

\begin{theorem}\label{thm:weight-chain}
For any chain cap $R$ and $\gamma \geq R-1$, there exists a graph $G$ such that the price of fairness of $G$ under the weighted fairness rule is bounded by {\small $\POF{\mathcal{M}}{u_{WF}} \geq \frac{R-1}{R}$}.
\end{theorem}

In the worst case, weighted fairness allows a price of fairness that approaches $1$ as the cycle and chain caps increase. The price of fairness also approaches $1$ as $\gamma$ increases.

\begin{theorem}
\label{thm:beta}
With no chain cap, there exists a graph $G$ such that the price of fairness of $G$ under the weighted fairness rule is bounded by {\small $\POF{\mathcal{M}}{u_{WF}} \geq \frac{\gamma}{\gamma+1}$}.
\end{theorem}

A similar result exists with cycles rather than chains.

\begin{theorem}
\label{thm:betacycle}
With no cycle cap there exists a graph $G$ such that the price of fairness of $G$ under the weighted fairness rule is bounded by {\small  $\POF{\mathcal{M}}{u_{WF}} \geq \frac{\gamma}{\gamma+1}$}.
\end{theorem}

These bounds show that weighted fairness allows for a price of fairness that approaches $1$, i.e., arbitrarily bad, as the cycle cap, chain cap, or $\gamma$ increase. 

We have shown that the worst-case prices of fairness approach $1$ under both the lexicographic and weighted fairness rules of~\englishcite{Dickerson14:Price}. Next, we propose a rule that favors disadvantaged groups, but also strictly \emph{limits} the price of fairness using a parameter set by policymakers.

\section{Hybrid Fairness Rule}\label{sec:combined}

In this section, we present a hybrid fair utility function that balances lexicographic fairness and a utilitarian objective. We generalize the hybrid utility function proposed by \englishcite{Hooker12:Combining}, which chooses between a Rawlsian (or maximin) objective and a utilitarian objective for multiple classes of agents. 

\subsection{Utilitarian and Rawlsian Fairness}\label{ssec:util-rawls}

Consider two classes of agents that receive utilities $u_1(X)$ and $u_2(X)$, respectively, for outcome $X$. The fairness rule introduced by~\englishcite{Hooker12:Combining} maximizes the utility of the worst-off class, unless this requires taking too many resources from other classes. When the inequality exceeds a threshold $\Delta$ (i.e., $|u_1(X)-u_2(X)|>\Delta$) they switch to a utilitarian objective that maximizes $u_1(X) + u_2(X)$. The utility function for this rule is

\begin{figure*}[ht!]
\centering
\begin{subfigure}[Hybrid Rawlsian-Utilitarian]{\includegraphics[width=.3\linewidth] {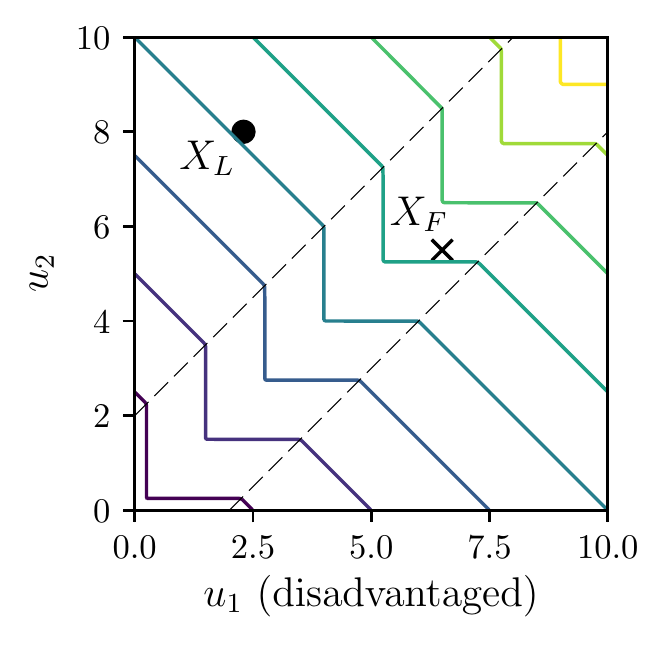}
   \label{fig:rawls}
 }%
\end{subfigure}\hfill
\begin{subfigure}[$\Delta$-Lexicographic ($u_{\Delta 1}$)]{\includegraphics[width=.3\linewidth] {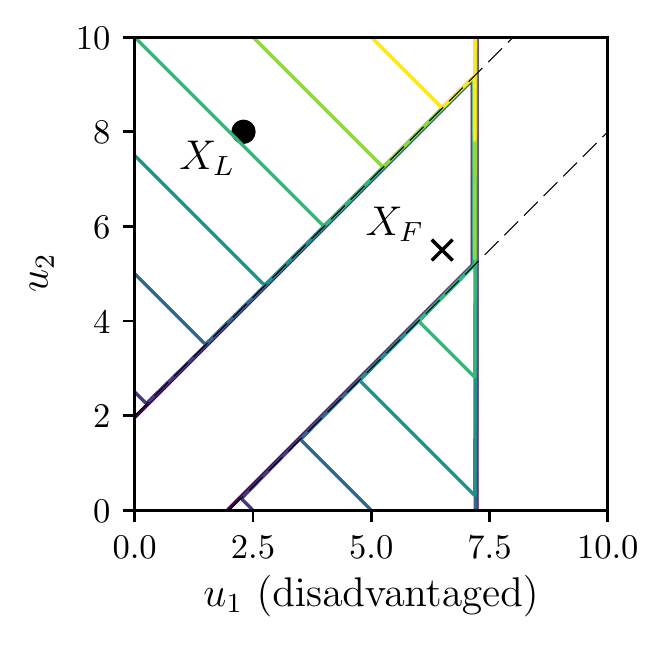}
   \label{fig:lex}
 }%
\end{subfigure}\hfill
 \begin{subfigure}[Relaxed $\Delta$-Lexicographic ($u_\Delta$)]{\includegraphics[width=.3\linewidth]{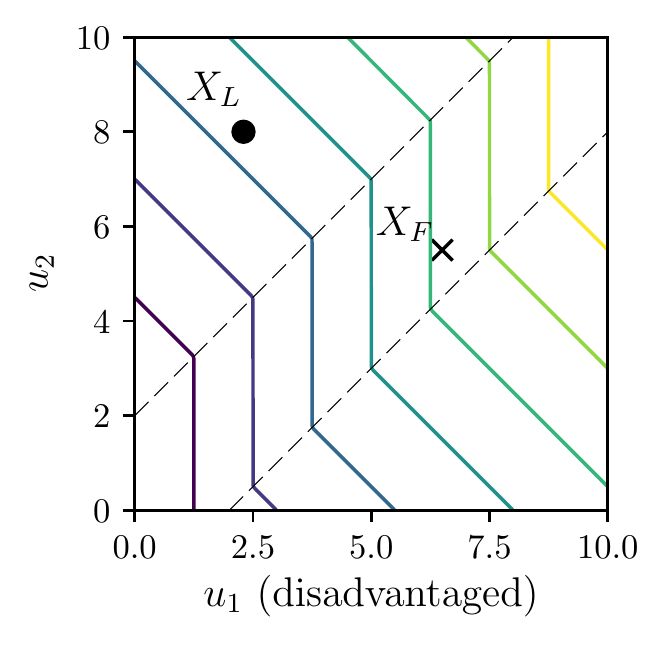}
   \label{fig:rlex}
 }%
\end{subfigure}%
\caption{Level sets for hybrid fair utility functions with $\Delta=2$, with example outcomes $X_L$ and $X_F$.}
\label{fig:contours}
\end{figure*}
{\small
\[
u_{\Delta}(X) =
\begin{cases}
2\min(u_1(X),u_2(X))+\Delta \\
\hspace{20pt}\text{if } |u_1(X)-u_2(X)| \leq \Delta \\
\\
u_1(X) + u_2(X) \\
\hspace{20pt}\text{otherwise.} 
\end{cases}
\]
}
The parameter $\Delta$ is problem-specific, and should be chosen by policymakers. Figure \ref{fig:rawls} shows the level sets of this utility function, with $\Delta =2$. This utility function can be generalized by switching to a different fairness rule in the \emph{fair region} (i.e. when $|u_1(X)-u_2(X)|\leq \Delta$). The next section generalizes this rule using lexicographic fairness.

\subsection{Hybrid-Lexicographic Rule}\label{sec:hybrid}

When it is desirable to favor one class of agents $g_1$ over class $g_2$, lexicographic fairness favors $g_1$. We propose a rule that implements lexicographic fairness only when inequality between groups does not exceed $\Delta$. This rule uses two steps: 1) determine whether inequality is small enough to use lexicographic fairness 2) choose the optimal outcome. These steps are outlined below, and formalized in Algorithm~\ref{alg1}.

\textbf{Step 1:} Find all outcomes that maximize a hybrid utility function, and determine whether lexicographic fairness is appropriate.

We use a utility function to identify outcomes that satisfy either a lexicographic or utilitarian objective. \Eqref{eq:hybrid1} shows one option for such a utility function, which assigns strict lexicographic utility ($\alpha=1$) according to \Eqref{eq:alex} in the fair region, and utilitarian utility otherwise. 
{\small
\begin{equation}\label{eq:hybrid1}
u_{\Delta 1}(X) =
\begin{cases}
u_1(X)+u_2(X)&\text{if } |u_1(X)-u_2(X)| \leq \Delta \\ 
 & \text{ and } u_1(X) = \max\limits_{X'\in \mathcal{X}}(u_1(X'))\\ 
 \\
u_1(X)+u_2(X) &\text{if } |u_1(X)-u_2(X)| >\Delta\\
\\
 0 &\text{otherwise.}
\end{cases}
\end{equation}
}
where $\mathcal{X}$ is the set of all possible outcomes. Figure~\ref{fig:lex} shows the contours $u_{\Delta 1}$. This utility function is clearly too harsh---it assigns zero utility to outcomes in the fair region that do not maximize $u_1$, and its optimal outcomes are not always Pareto efficient. Consider outcomes $X_F$ and $X_L$ in Figure~\ref{fig:lex}. $X_F$ is in the fair region but does not maximize $u_1$, so $u_{\Delta 1}(X_F)=0$; $X_L$ is in the utilitarian region but is less efficient, so $u_{\Delta 1}(X_L)= u(X_L)$. Under utility function $u_{\Delta 1}$, the less-efficient outcome $X_L$ is chosen over $X_F$.

To address this problem we introduce $u_{\Delta}$ in \Eqref{eq:rlex}, which relaxes $u_{\Delta 1}$. For outcomes in the fair region (that is, with $|u_1-u_2|\leq \Delta$), utility is assigned proportional to $u_1$. As shown in Figure~\ref{fig:rlex}, the contours of $u_{\Delta }$ are continuous. 

{\small
\begin{equation}
\label{eq:rlex}
u_{\Delta }(X) =
\begin{cases}
u_1(X)+u_2(X)-\Delta &\text{if } u_2(X)-u_1(X)>\Delta\\
2u_1(X) &\text{if } |u_1(X)-u_2(X)|\leq \Delta\\
u_1(X)+u_2(X)+\Delta &\text{if } u_1(X)-u_2(X)>\Delta
\end{cases}
\end{equation}
}

Let $X_{\mathit{OPT}}$ be the set of outcomes that maximize $u_{\Delta}$. If any outcomes in $X_{\mathit{OPT}}$ are in the utilitarian region , then any utilitarian-optimal outcome is selected. However, if any outcomes in $X_{\mathit{OPT}}$ are in the fair region, then Step 2 must be used. This process is described below, and formalized in Algorithm~\ref{alg1}.


\textbf{Step 2:} If any solution in $X_{\mathit{OPT}}$ is in the fair region, select the lexicographic-optimal solution in the fair region.

The utility function $u_{\Delta }$ assigns the same utility to all solutions in the fair region with the same $u_1(X)$, no matter the value of $u_2(X)$. However, if there exist two outcomes $X_A$ and $X_B$ such that $u_1(X_A)=u_1(X_B)$ and $u_2(X_A)>u_2(X_B)$, then $X_A$ is lexicographically preferred to $X_B$. 

{\small
\begin{algorithm}[H]                  
\caption{FairMatching}          
\label{alg1}                           
\textbf{Input:} Threshold $\Delta$, matchings $\Mcal$ \\
\textbf{Output:} Fair matching $M^*$ 
\begin{algorithmic}       
\State $\Mcal_{OPT} \leftarrow \arg\max_{M \in \mathcal{M}} u_{\Delta }(M)$
\If{$|\Mcal_{OPT}|>1$}
\State Select a matching $M\in \Mcal_{OPT}$
\If{$M$ is in the utilitarian region}
\State $M^* \leftarrow M$
\Else
\State $\Mcal_1 \leftarrow \{ M' \in \Mcal_{OPT} \mid u_1(M')=u_1(M) \}$
\State $M^* \leftarrow \arg\max_{M' \in \Mcal_1} u_2(M')  $
\EndIf
\Else
\State $M^* \leftarrow \Mcal_{OPT}$
\EndIf
\end{algorithmic}
\end{algorithm}
}

\subsection{Hybrid Rule for Several Classes}

We now generalize the hybrid-lexicographic fairness rule to more than two classes. Consider a set $\Pcal{}$ of classes $g_i$, $i=1,\dots,|\Pcal{}|$. Let there be an ordering $\succ$ over $g_i$, where $g_a \succ g_b$ indicates that $g_a$ should receive higher priority over $g_b$. WLOG, let the preference ordering over $g_i$ be $g_1 \succ g_2 \succ \dots \succ g_P$. Let $u_i(X)$ be the utility received by group $i$ under outcome $X$. As in the previous section, we 1) use a utility function to determine whether lexicographic fairness is appropriate, then 2) select either a lexicographic- or utilitarian-optimal outcome. 

\textbf{Step 1:} To define a utility function, we observe that in \Eqref{eq:rlex}, in the utilitarian region a positive offset $\Delta$ is added if $u_1(X)>u_2(X)$, and a negative offset is added otherwise. With $|\Pcal{}|$ classes, each solution in the utilitarian region receives a utility offset of $+\Delta$ if $u_1(X)>u_i(X)$, and $-\Delta$ otherwise, for each class $i=2,3,\dots,|\Pcal{}|$. As in the previous section, these offsets ensure continuity in the utility function, and ensure that at least one of the maximizing outcomes will be Pareto optimal.


{\small
\begin{equation}
\label{geneq}
u_{\Delta }(X) =  
\begin{cases}
|\Pcal{}|\cdot u_1(X) \\
\hspace{20pt} \text{if } \max_i(u_i(X))-\min_i(u_i(X))\leq \Delta,  \\
\\
u_1(X)+\sum_{i=2}^{|\Pcal{}|} (u_i(X) + \sgn(u_1(X)-u_i(X) )\Delta) \\
\hspace{20pt} \text{otherwise } \\ 
\end{cases}
\end{equation}
}

\textbf{Step 2:} Let $X_{\mathit{OPT}}$ be the set of solutions that maximize $u_{\Delta}$. If all optimal solutions are in the utilitarian region, any utilitarian-optimal solution is selected. If any optimal solution is in the fair region, then the lexicographic-optimal solution in the fair region must be selected, subject to the preference ordering $g_1 \succ g_2 \succ \dots \succ g_{|\Pcal{}|}$. 

{\small
\begin{algorithm}[H]
\caption{FairMatching for $|\Pcal{}| \geq 2$ classes}          
\label{alg2}                           
\textbf{Input:} Threshold $\Delta$, matchings $\Mcal$ \\
\textbf{Output:} Fair matching $M^*$ 
\begin{algorithmic}       
\State $\Mcal_{OPT} \leftarrow \arg\max_{M \in \Mcal} u_{\Delta }(M)$
\If{$|\Mcal_{OPT}|>1$}
\State Select a matching $M\in M_{OPT}$
\If{$M$ in utilitarian region}
\State $M^* \leftarrow M$
\Else 
\State $\Mcal_1 \leftarrow \{ M' \in \Mcal_{OPT} \mid u_1(M')=u_1(M) \}$
\For{$i=2,\dots,|\Pcal{}|$}
\State $\Mcal_i \leftarrow \arg\max_{M' \in \Mcal_{i-1}} u_i(M') $
\EndFor
\State $M^* \leftarrow$ any matching in $\Mcal_{|\Pcal{}|}$
\EndIf
\Else
\State $M^* \leftarrow \Mcal_{OPT}$
\EndIf
\end{algorithmic}
\end{algorithm}
}

\subsection{Price of Fairness for the Hybrid-Lexicographic Rule}

Theorem~\ref{thm:hybridpof} gives a bound on the price of fairness for the hybrid-lexicographic rule; its proof is given in \Appref{sec:poffairness}.

\begin{theorem}
\label{thm:hybridpof}
Assume the optimal utilitarian outcome $X_E$ receives utility $u(X_E)=u_E$, with most prioritized class $g_1 \in \Pcal{}$ receiving utility $u_1$, and $Z$ other classes $g_i \in \Pcal{}$ such that $u_1(X_E)>u_i(X_E)$. Then, $ \POF{\mathcal{M}}{u_{\Delta }} \leq \frac{2((|\Pcal{}|-1)-Z)\Delta}{u_E}$.
\end{theorem}

\subsection{Hybrid Fairness in Kidney Exchange}

The hybrid-lexicographic fairness rule in \Eqref{eq:rlex} is easily applied to kidney exchange, with $u_H$ and $u_L$ the total utility received by highly-sensitized and lowly-sensitized patients, respectively,

{\small
\begin{equation}
u_{\Delta}(M)=
\begin{cases}
u_L(M)+u_H(M)-\Delta &\text{if } u_L(M)-u_H(M)>\Delta\\
2u_H(M) &\text{if } |u_L(M)-u_H(M)|\leq \Delta\\
u_L(M)+u_H(M)+\Delta &\text{if } u_H(M)-u_L(M)>\Delta
\end{cases}
\label{utilfunction}
\end{equation}
}

In the following section, we demonstrate the practical effectiveness of the hybrid-lexicographic rule by testing it on real kidney exchange data.

\section{Experiments}\label{sec:experiments}
In this section, we compare the behavior of $\alpha$-lexicographic, weighted, and hybrid-lexicographic fairness. All code for these experiemnts are available on GitHub.\footnote{https://github.com/duncanmcelfresh/FairKidneyExchange} We use each rule to find the optimal fair outcomes for $314$ real kidney exchanges from the United Network for Organ Sharing (UNOS), collected between 2010 and 2016. To solve the kidney exchange clearing problem (KEP) we use the PICEF formulation introduced by \englishcite{Dickerson16:Position}, with cycle cap $3$ and various chain caps. In real exchanges, not all recommended edges in a matching result in successful transplants. To reflect this uncertainty, we use the concept of failure-aware kidney exchange introduced in \cite{Dickerson13:Failure}: all edges in the exchange can fail with probability $(1-p)$; the matching algorithm maximizes \textit{expected} matching weight, considering edge success probability $p$.

\subsection{Procedure}

For each UNOS exchange graph $G$, we use the following procedure to implement each fairness rule. We repeat the following procedure for chain caps $0$, $3$, $10$, and $20$, and for edge success probabilities $p=0.1 n$, with $n=1,2,\dots,10$.

\begin{enumerate} 
\item Find the efficient matching $M_E$ by solving the to optimality the NP-hard kidney exchange problem (KEP) on $G$. 
\item Find the fair matching $M_F$ by solving the KEP on $G'=(V,E')$, where each edge $e \in E'$ has weight 1 if $e$ ends in $V_H$ and 0 otherwise.
\item \textbf{Weighted Fairness:} Find the $\gamma$-fair matching $M_\gamma$ by solving the KEP on $G^\gamma=(V,E^\gamma)$, where each edge $e\in E^\gamma$ has weight $1+\gamma$ if $e$ ends in $V_H$ and 1 otherwise. After finding $M_\gamma$, the reported utilities are calculated using edge weights of $E$ and not $E'$. We use weight parameters $\gamma = 2 n$, with $n=0,1,2,\dots,10$.
\item \textbf{$\alpha$-Lexicographic Fairness:} Find the $\alpha$-fair matching $M_\alpha$ by solving the KEP on $G$, with the additional constraint $u_H(M_\alpha)\geq \alpha u_H(M_E)$. We use parameters $\alpha = 0.1 n$, with $n=0,1,2,\dots,10$.
\item \textbf{Hybrid-Lexicographic Fairness:} Find the $\Delta$-fair matching $M_\Delta$ using the $\alpha$-fair matchings $M_\alpha$, and Algorithm \ref{alg1}. That is, $M_\Delta=\mathrm{FairMatching}(\Delta,M_\alpha)$. We use parameters $\Delta = 0.1 n\cdot u(M_E)$, with $n=0,1,2,\dots,10$.
\end{enumerate}

\begin{figure*}[ht!]
\centering
\includegraphics[width=\textwidth] {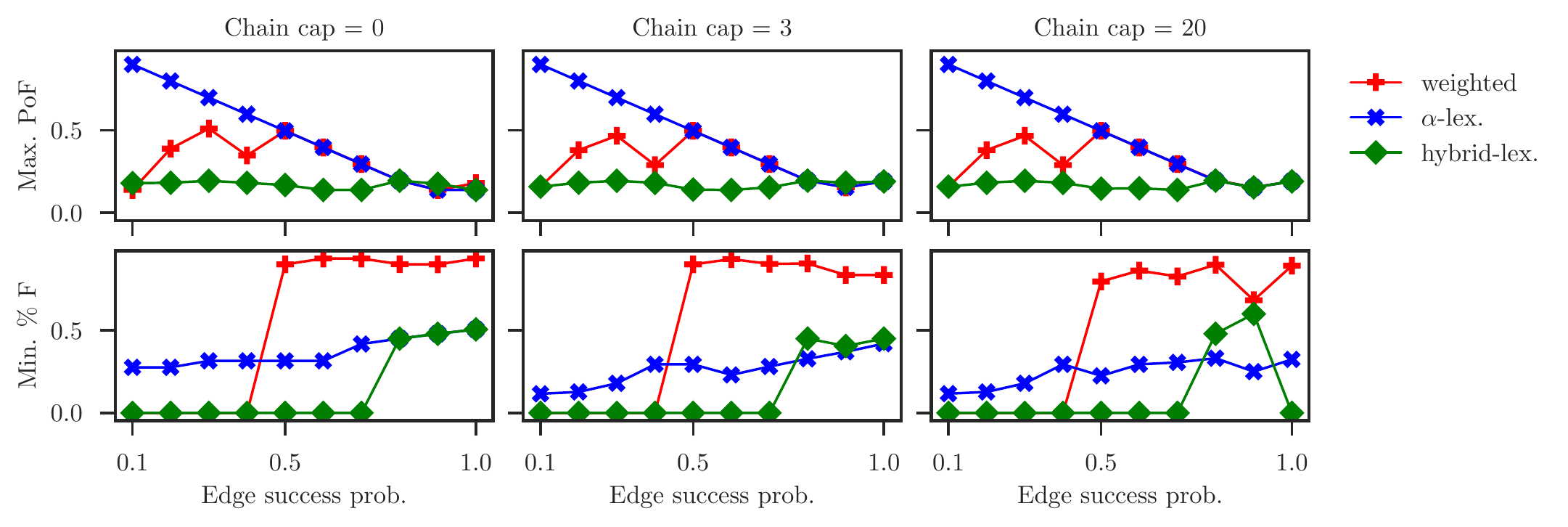}
\caption{Worst-case price of fairness and $\% F$ for various edge success probabilities, and fairness parameters $\alpha=0.1$, $\gamma=0.1$, $\Delta = 0.1 u(M_E)$. }
                    
   \label{fig:pof-minparam}
\end{figure*}

Throughout this procedure, we calculate the utility of the efficient matching ($u_E$) and the fair matching ($u_F$) for each UNOS graph, and for each fairness rule---with parameters $\alpha \in [0,1]$, $\gamma \in [0,20]$, and $\Delta \in [0, u(M_E)]$. 


There are two important outcomes of each fairness rule: Price of Fairness (PoF), and fraction of the fair score ($\% F$). To calculate PoF we use the definition in \Eqref{eq:pof}, using $u_E$ and $u_F$. We define $\% F$ as the fraction of the maximum highly sensitized utility, achieved by $M_{\{\alpha,\gamma,\Delta\}}$, defined as
$$\% F(M_{\{\alpha,\gamma,\Delta\}}, M_F) = u_H(M_{\{\alpha,\gamma,\Delta\}}) / u_H(M_F).$$
PoF and $\% F$ indicate the efficiency loss and the fairness of each rule, respectively. 

\subsection{Results and Discussion}

Each fairness rule offers a parameter that balances efficiency and fairness. Two of these rules guarantee a certain outcome: $\alpha$-lexicographic guarantees fairness, but allows high efficiency loss, while hybrid-lexicographic bounds overall efficiency loss.  Weighted fairness makes no guarantees. 

The price of fairness can be high in real exchanges, especially when edge success probability $p$ is small. In failure-aware kidney exchange, cycles and chains of length $k$ receive utility proportional to $p^k$. Fair matchings often use longer cycles and chains than the efficient matching, in order to reach highly sensitized patients; this leads to a high price of fairness when $p$ is small.



Even when $\alpha$ and $\gamma$ are small, there are cases when both $\alpha$-lexicographic and weighted fairness allow for a high PoF. This becomes worse with lower edge probability. Figure~\ref{fig:pof-minparam} shows the worst-case PoF and $\% F$ for each rule, for the smallest parameters tested, for a range of edge success probabilities. \Appref{sec:appendixresults} contains results for all parameter values tested.

Hybrid-lexicographic fairness limits PoF within the guaranteed bound of $0.2$; this comes at the cost of a low $\% F$---when edge success probability is small, hybrid-lexicographic fairness awards zero fair utility in the worst case. $\alpha$-lexicographic fairness produces the opposite behavior: $\% F$ is always larger than the guaranteed bound of $0.1$, but the worst-case price of fairness grows steadily as edge probability decreases.

Theory suggests that the price of fairness is small on denser random graphs (see Section~\ref{sec:randomgraph}). We empirically confirm this theoretical finding by calculating the worst-case price of fairness and $\% F$ for random graphs of various sizes generated from real data; these results are given in 
\Secref{sec:appendixresults}. In this case---when the price of fairness is small---$\alpha$-lexicographic fairness may be appropriate, as overall efficiency loss is not severe.

Both $\alpha$-lexicographic and hybrid-lexicographic fairness are useful, depending on the desired outcome. Policymakers may choose between these rules, and set the parameters $\alpha$ and $\Delta$ to guarantee either a minimum $\% F$ or a maximum price of fairness. 


\section{Conclusion}\label{sec:conclusion}
We addressed the classical problem of balancing fairness and efficiency in resource allocation, with a specific focus on the kidney exchange application area. Extending work by~\englishcite{Ashlagi14:Free} and~\englishcite{Dickerson14:Price}, we show that the theoretical price of fairness is small on a random graph model of kidney exchange, when both cycles and chains are used. However this model is too optimistic---real kidney exchanges are less certain and more sparse, and in reality the price of fairness can be unacceptably high.

Drawing on work by~\englishcite{Hooker12:Combining}, which is not applicable to kidney exchange, we provided the first approach to incorporating fairness into kidney exchange in a way that prioritizes marginalized participants, but also comes with acceptable worst-case guarantees on overall efficiency loss.  Furthermore, our method is easily applied as an objective in the mathematical-programming-based clearing methods used in today's fielded exchanges. Using data from a large fielded kidney exchange, we showed that our method bounds efficiency loss while also prioritizing marginalized participants when possible.

Moving forward, it would be of theoretical and practical interest to address fairness in a realistic \emph{dynamic} model of a matching market like kidney exchange~\cite{Anshelevich13:Social,Akbarpour14:Dynamic,Anderson15:Dynamic,Dickerson15:FutureMatch}.  For example, how does prioritizing a class of patients in the present affect their, or other groups', long-term welfare?  Similarly, exploring the effect on long-term efficiency of the single-shot $\Delta$ we use in this paper would be of practical importance; to start, $\Delta$ can be viewed as a hyperparameter to be tuned~\cite{Thornton13:Auto}.

\bibliographystyle{named}
\bibliography{refs,bib}

\clearpage
\appendix
\section{Price of Fairness in the Random Graph Model}\label{sec:pofrandomgraph}

\englishcite{Ashlagi14:Free} characterize efficient matchings in a random graph model without chains, and \englishcite{Dickerson14:Price} build on this to show that the price of fairness without chains is bounded above by $2/33$. \englishcite{Dickerson12:Optimizing} extend the efficient matching of~\englishcite{Ashlagi14:Free} to include chains, but do not calculate the price of fairness.  We close the remaining theory gap regarding the price of fairness with chains.  \Appref{sec:randgraphmodel} describes the random graph model, and \Appref{sec:pofwithchains} presents the theoretical price of fairness with chains.

\subsection{Random Graph Model}\label{sec:randgraphmodel}

Let all patient-donor pairs $P$ be partitioned into subsets $V^\XY$ for each patient blood type $X$ and donor blood type $Y$. These subsets will be further partitioned into lowly- and highly sensitized pairs $V^\XY_L$ and $V^\XY_H$. Let $\mu_X$ be the fraction of both patients and donors of each blood type $X$.  

Let $N^X$ be the set of NDDs of blood type $X$. Let $\beta |P|$ be the total number of NDDs, with the same blood type distribution as patients. That is, $|N^X|=\beta \mu_X  |P|$, with $X\in \{A,B,AB,O\}$.

Patient-donor vertices may be blood-type compatible, but will not be connected by a directed edge due to tissue-type incompatibility. Let $\pbar$ be the fraction of patient-donor pairs that are blood-type-compatible, but tissue-type-incompatible.

We refer to certain blood-type vertex subsets of as follows:

\begin{enumerate}
\item $V^\AB$ and $V^\BA$: reciprocal pairs
\item $V^\XX$: self-demanded pairs
\item $V^\ABB,V^\ABA,V^\ABO,V^\AO,V^\BO$: over-demanded pairs
\item $V^\AAB,V^\BAB,V^\OA,V^\OB,V^\OAB$: under-demanded pairs
\end{enumerate}

To reflect real-world exchanges, assume $\pbar > 1-\lambda$,  $\mu_\O>\mu_\A>\mu_\B>\mAB$, and $\pbar<2/5$. WLOG, let $|V^\AB|>|V^\BA|$, and assume that the absolute difference between these pools grows sublinearly with the size of the exchange, that is $|V^\AB|-|V^\BA|=o(n)$.

\subsection{The Price of Fairness With Chains} \label{sec:pofwithchains}

We calculate the price of fairness in this model by exploring all of the possible ways that the efficient matching can proceed, which depends on $\beta$. We state without proof that there are only four possible matchings with nonzero price of fairness, and several matchings with zero price of fairness. It is tedious, but straightforward, to confirm this statement, using the assumptions made while constructing these matchings. Figure \ref{fig:matchings} shows each possible matching on this model, and some of the impossible matchings.
 
\begin{figure*}[ht!]
\centering
\includegraphics[width=\linewidth] {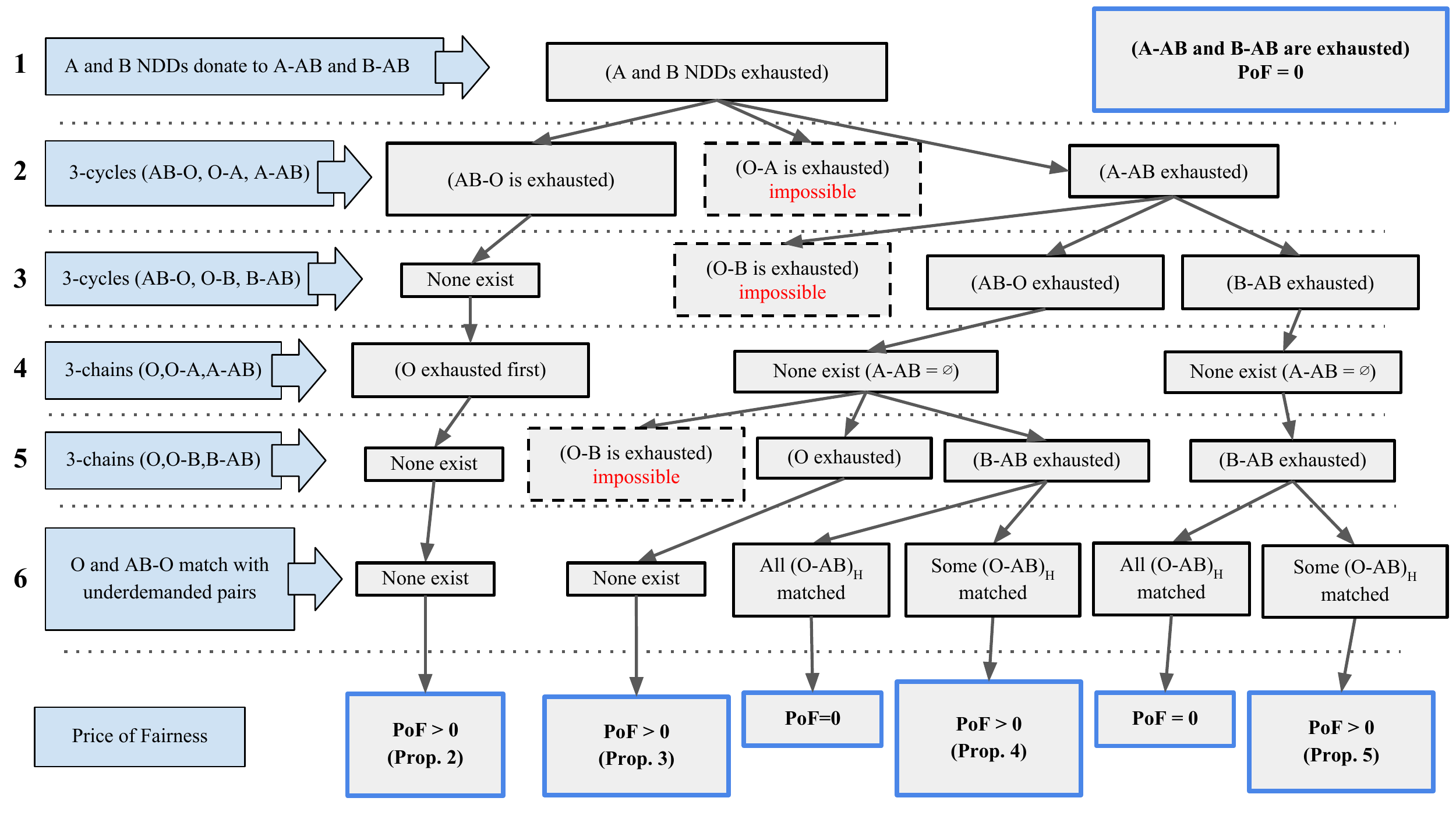}
\caption{All possible matchings on the random graph model. Boxes with blue borders represent the matching outcomes, and boxes with black borders represent intermediate steps in each matching. Some of the impossible matchings are shown as boxes with dashed black borders.}
\label{fig:matchings}
\end{figure*} 

Propositions \ref{prop:smallbeta}, \ref{prop:midbeta11}, \ref{prop:midbeta13}, and \ref{prop:midbeta3} give the price of fairness for each of the four matchings with nonzero price of fairness; for each of these cases, $\beta<\mAB(1-\pbar)$. Proposition \ref{prop:largebeta} states that the price of fairness is zero when $\beta>\mAB(1-\pbar)$.

In all of these matchings, the price of fairness is bounded above by the price of fairness without NDDs, found by \englishcite{Dickerson14:Price}; Theorem \ref{thm:pofdecreases} states this finding, which uses by Lemmas \ref{lem:prop12-bounded} and \ref{lem:prop34-bounded}.

Theorem \ref{thm:maxbeta} states that the price of fairness is zero when $\beta>1/8$, and Lemmas \ref{lem:maxbeta-small}, \ref{lem:maxbeta-mid11}, \ref{lem:maxbeta-mid13}, and \ref{lem:maxbeta-mid3} give bounds on $\beta$ for each matching with nonzero price of fairness.

We start with the efficient matching proposed in~\cite{Dickerson12:Optimizing} using cycles and chains up to length 3. This matching may proceed in many different ways, depending on $\beta$. However, most outcomes are impossible based on the canonical assumptions for the random graph model. Figure \ref{fig:matchings} shows all possible ways that the matching can proceed.

Lemma \ref{lem:underdemand} states that even without chains, all highly-sensitized patients except for those in $V^\OAB$ are matched in the efficient matching, only using cycles; this Lemma will be used in all following propositions.

\begin{lemma}\label{lem:underdemand}

Denote by $\Mcal{}$ the set of matchings in $G(n)$ using cycles and chains up to length 3. As $n \rightarrow \infty$, a.s. all highly sensitized pairs can be matched with no efficiency loss under the lexicographic fairness rule, except for those of type $\OAB$.\\
\\
(This Lemma uses the same efficient matching introduced by Dickerson~\cite{Dickerson12:Optimizing}.)
\end{lemma}

\begin{proof}[sketch] 

Begin with the efficient matching $M^*$ using only cycles up to length 3, proposed by Dickerson in~\cite{Dickerson14:Price}. $M^*$ matches all over-demanded and self-demanded vertices with high probability, but leaves some under-demanded vertices unmatched. We proceed through the initial steps of matching $M^*$ to show that \textit{all} vertices in $V_H^{\OA}$, $V_H^{\OB}$, $V_H^{\AAB}$, and $V_H^{\BAB}$ are matched.

\begin{enumerate}
\item Match all vertices in $V^\BA$ in 2-cycles with $V^\AB$, exhausting $V^\BA$ and leaving $|V^\AB|\propto o(n)$. 

\item Match all remaining vertices in $V^\AB$ in 3-cycles with $V^\BO$ and $V^\OA$. There are only $|V^\AB|\propto o(n)$ of these cycles, which will become negligible to the price of fairness as $n \rightarrow \infty$.

\item Match all remaining vertices in $V^\AO$ in 2-cycles with $V^\OA$. Note that $|V^\AO|\propto \pbar\mu_A \mu_O$ and $|V^\OA| \propto \mu_A \mu_O$. The $V^\AO$ vertices are exhausted first if $|V^\AO|<|V^\OA|$, which holds almost surely because $\pbar\mu_A \mu_O<\mu_A \mu_O $ due to the assumption  $\pbar<2/5$. All highly sensitized vertices $V_H^\OA$ are matched because $(1-\lambda) \mu_A \mu_O< \pbar\mu_A \mu_O  $ holds under the assumption $1-\lambda< \pbar $. Thus both $V^\AO$ and $V_H^\OA$ are exhausted, and $|V^\OA| \propto \mu_A \mu_O(1-\pbar)$.  

\item Match all remaining vertices in $V^\BO$ in 2-cycles with $V^\OB$. Note that $|V^\BO|\propto \pbar\mu_B \mu_O$ and $|V^\OB| \propto \mu_B \mu_O$. As before, the a.s. $|V^\OB|>|V^\BO|$. All highly sensitized vertices $V_H^\OB$ are matched a.s., because $\pbar\mu_B \mu_O > (1-\lambda) \mu_B \mu_O $ holds under the assumption $\pbar > 1-\lambda$. Thus both $ V^\BO$ and $V_H^\OB$ are exhausted, and $|V^\OB|\propto \mu_B \mu_O(1-\pbar)$.

\item Match all vertices in $V^\ABA$ in 2-cycles with $V^\AAB$. Note that, $|V^\ABA|\propto \pbar\mu_A \mAB$ and $|V^\AAB| \propto \mu_A \mAB$. As before, a.s. $|V^\AAB|>|V^\AAB|$. All highly sensitized vertices $V_H^\AAB$ are matched, because $\pbar\mu_A \mAB > (1-\lambda) \mu_A \mAB $ under the assumption $\pbar > 1-\lambda$. Thus both $ V^\ABA$ and $V_H^\AAB$ are exhausted, and $|V^\AAB| \propto \mu_A \mu_O(1-\pbar)$.

\item Match all vertices in $V^\ABB$ in 2-cycles with $V^\BAB$. Note that, $|V^\ABB|\propto \pbar\mu_B \mAB$ and $|V^\BAB| \propto \mu_B \mAB$, and a.s. $|V^\BAB|>|V^\ABB|$. All highly sensitized vertices $V_H^\BAB$ are matched, because $\pbar\mu_B \mAB > (1-\lambda) \mu_B \mAB $ under the assumption $\pbar > 1-\lambda$. Thus both $ V^\ABB$ and $V_H^\BAB$ are exhausted, and $|V^\BAB| \propto \mu_B \mu_O(1-\pbar)$.

\end{enumerate}

Thus, these initial steps of matching $M^*$ exhaust all highly sensitized pairs in $V_H^{\OA}$, $V_H^{\OB}$, $V_H^{\AAB}$, and $V_H^{\BAB}$.

\end{proof}

With uniform edge weights, lexicographic fairness requires that we match the maximum possible number of highly sensitized vertices. Lemma \ref{lem:underdemand} states that the efficient matching $M^*$ includes all highly sensitized patients, except for those in $V^{\OAB}$. Therefore all efficiency loss--and price of fairness--is caused by matching vertices in $V_H^{\OAB}$.

Using both chains and cycles increases overall efficiency. In the dense graph model used in this Appendix, adding chains can only decrease the price of fairness. 

Proposition 1 in \cite{Dickerson14:Price} states that with only cycles up to length 3, and assuming $\pbar > 1-\lambda$, and $\mu_\O<3\mu_\A/2$, and $\mu_\O>\mu_\A>\mu_\B>\mAB$, the price of fairness is at most $\frac{2}{33}$. In the dense graph model used here, adding chains tightens this upper bound. 

The following propositions tighten the upper bound on the price of fairness, for every possible value of $\beta$.

\begin{proposition}\label{prop:largebeta}
Assume
\begin{enumerate}[label=\textbf{\arabic*}]
\item $\beta>(1-\pbar)\mAB$.
\end{enumerate}

Denote by $\Mcal{}$ the set of matchings in $G(n)$ using cycles and chains up to length 3. As $n \rightarrow \infty$, almost surely $\POF{\mathcal{M}}{u_{LEX}}=0$.
\end{proposition}

\begin{proof}[sketch]

We begin by executing the initial steps of matching $M^*$ as done in the proof of Lemma \ref{lem:underdemand}, matching all highly sensitized vertices except for those in $V_H^\OAB$. The following steps continue the matching $M^*$ from Lemma \ref{lem:underdemand}.

\begin{enumerate}\setcounter{enumi}{6}

\item $\A$- and $\B$-type NDDs donate to $V^\AAB$ and $V^\BAB$, respectively. Note that $|N^A|\propto \beta \mu_A$ and $|V^\AAB|\propto  (1-\pbar)\mu_A\mAB$. Assuming $\beta >(1-\pbar)\mAB$, the inequality $\beta \mu_\A > \mu_A\mAB(1-\pbar)$ holds and a.s. $|N^\A|>|V^\AAB|$. By the same argument, a.s. $|N^\B|>|V^\BAB|$. Thus, both $V^\AAB$ and $V^\BAB$ are exhausted, and $|N^\B| \propto \mu_\B \left( \beta - (1-\pbar)\mAB \right) $ and $|N^\A| \propto \mu_\A \left( \beta - (1-\pbar)\mAB \right) $.

\item Create cycles of the form ($\ABO$, $\OX$, $\XAB$), with $\X\in \{\A,\B\}$. None of these cycles occur because both $V^\AAB$ and $V^\BAB$ have been exhausted in previous steps.

\item Create chains of the form  ($\O$,$\OX$,$\XAB$), with $\X\in \{\A,\B\}$. None of these cycles occur, because both $V^\AAB$ and $V^\BAB$ have been exhausted in previous steps.

\item Remaining $\O$-type NDDs donate to remaining under-demanded vertices, beginning with $V^\OAB$. Note that no $\O$-type NDDs have been used in previous steps, so $|N^\O| \propto \beta \mu_\O$.

\item 2-cycles are created with $V^\ABO$ and remaining under-demanded vertices, beginning with $V^\OAB$. Note that no vertices in $V^\ABO$  have been used in previous steps, so $|V^\ABO| \propto \pbar\mu_\O \mAB$.

\end{enumerate}

The final two steps match up to $|V^\ABO| + |N^\O| \propto \beta \mu_\O + \pbar\mu_\O \mAB$ vertices in $V^\OAB$. The only remaining highly-sensitized vertices are in $V_H^|OAB\propto(1-\lambda)  \mu_\O \mAB$. Assuming that $\pbar>1-\lambda$, the inequality $\beta \mu_\O + \pbar\mu_\O \mAB >   \pbar\mu_\O \mAB > (1-\lambda) \mu_\O \mAB$ holds, and a.s. $|V^\ABO| + |N^\O|>|V_H^|OAB|$. This exhausts all vertices in $|V_H^\OAB|$. All other highly-sensitized vertices were matched in steps 1-6 of, as in Lemma \ref{lem:underdemand}. Thus, all highly sensitized vertices can be matched with no efficiency loss, and the price of fairness is zero.

\end{proof}

Proposition\ref{prop:largebeta} assumes that $\beta$ is extremely large, specifically  $\beta >1/4>(1-\pbar)\mAB$. In practice, $\beta<0.01$ -- that is, the number of NDDs in an exchange is often less than $1 \%$  of the size of the exchange. The following Propositions address the price of fairness when $\beta < (1-\pbar)\mAB<1/4$.

\begin{proposition}\label{prop:smallbeta}
Assume
\begin{enumerate}[label=\textbf{A.\arabic*}]
\item  $\beta< \mu_\A \left(1-\pbar\right) - \pbar \mu_\text{AB}$ \label{smallb:b1}
\item $\beta < \mAB (1 - \pbar ) - \pbar \mAB \mu_\O / \mu_\A$ \label{smallb:b1a}
\item $\beta < \mAB \left(\frac{\mu_\A}{\mu_\A+\mu_\O} - \pbar \right)$ \label{smallb:b1b}
\end{enumerate}

These constraints imply $\beta \in [0,1/8]$.Denote by $\Mcal{}$ the set of matchings in $G(n)$ using cycles and chains up to length 3. Almost surely as $n \rightarrow \infty$, the price of fairness is

\begin{align*}
\POF{\mathcal{M}}{u_{LEX}} &= \frac{(1-\lambda)\mu_\O \mu_\text{AB}}{ u_E  }
\end{align*}

with 

\begin{align*}
u_E &=  \pbar \Big[ 2\mAB \mu_\B +  2\mAB \mu_\A +  3\mAB \mu_\O \\
 +& 2\mu_\A\mu_\O +  2\mu_\B \mu_\O + \mu_\O^2 + \mu_\A^2 + \mu_\B^2+ \mAB^2 \Big] \\
+& 2\mu_\A \mu_\B + \beta \left(\mu_\A + \mu_\B + 2\mu_\O \right) 
\end{align*}
\end{proposition}

\begin{proof}[sketch]

We begin with matching $M^*$ as done in the proof of Lemma \ref{lem:underdemand}, matching all highly sensitized vertices except for those in $V_H^\ABO$. We now complete the efficient matching using both 3-cycles and 3-chains as in~\cite{Dickerson12:Optimizing}.

\begin{enumerate}\setcounter{enumi}{6}

\item $\A$- and $\B$-type NDDs donate to $V^\AAB$ and $V^\BAB$, respectively. Note that $|N^A|\propto \beta \mu_A$ and $|V^\AAB|\propto  (1-\pbar)\mu_A\mu_\text{AB}$. The inequality $\beta \mu_\A < \mu_A\mu_\text{AB}(1-\pbar)$ holds due to assumption \ref{smallb:b1a},  and a.s. $|N^\A|<|V^\AAB|$. By the same argument, a.s. $|N^\B|<|V^\BAB|$. Thus, both $N^\A$ and $N^\B$ are exhausted, and $|V^\AAB| \propto \mu_\A \mu_\text{AB} (1-\pbar) - \beta \mu_\A $ and $|V^\BAB| \propto \mu_\B \mu_\text{AB} (1-\pbar) - \beta \mu_\B $.

\item Create cycles of the form ($\ABO$, $\OA$, $\AAB$). The current size of each vertex group is

\begin{enumerate}[label={(\arabic*)}]
\item $|V^\ABO| \propto \pbar \mu_\text{AB} \mu_\O$
\item $|V^\OA| \propto (1-\pbar )\mu_\A \mu_\O$
\item $|V^\AAB| \propto \mu_\A \mu_\text{AB} (1-\pbar ) - \beta \mu_\A$
\end{enumerate}

The inequality $(1)<(2)$ holds due to the model assumptions, so a.s. $|V^\ABO| < |V^\OA|$. Note that the inequality $(1)<(3)$ can be written as

$$\beta < \mAB (1 - \pbar ) - \pbar \mAB \mu_\O / \mu_\A$$

which holds by assumptions \ref{smallb:b1a}, and a.s. $|V^\ABO| < |V^\AAB|$. Executing these cycles exhausts $V^\ABO$ and leaves the following vertices remaining

\begin{enumerate}[label={(\arabic*)}]
\item $|V^\OA| \propto (1-\pbar)\mu_\A \mu_\O - \pbar \mAB \mu_\O$.
\item $|V^\AAB| \propto  (1 - \pbar )\mu_\A \mAB - \pbar \mAB \mu_\O -\beta \mu_\A$
\end{enumerate}

\item Create cycles of the form ($\ABO$, $\OB$, $\BAB$). The previous step exhausted $V^\ABO$, so none of these cycles occur.

\item Create chains of the form  ($\O$,$\OA$,$\AAB$). The current size of each vertex group is

\begin{enumerate}[label={(\arabic*)}]
\item $|N^\O| \propto \beta \mu_\O $
\item $|V^\OA| \propto (1-\pbar)\mu_\A \mu_\O - \pbar \mAB \mu_\O $
\item $|V^\AAB| \propto  (1 - \pbar )\mu_\A \mAB - \pbar \mAB \mu_\O -\beta \mu_\A $
\end{enumerate}

The inequality $(1)<(2)$ holds due to assumption \ref{smallb:b1}, so a.s. $|N^\O| <|V^\OA|$. Note that inequality $(1)<(3)$ can be written as 

$$\beta < \mAB \left(\frac{\mu_\A}{\mu_\A+\mu_\O} - \pbar \right)$$

which holds due to \ref{smallb:b1b}. Thus a.s. $|N^\O|<|V^\AAB| $, and $|N^\O|$ is exhausted. The vertices unmatched by these chains are

\begin{enumerate}[label={(\arabic*)}]
\item $|V^\OA| \propto (1-\pbar)\mu_\A \mu_\O - \pbar \mu_\text{AB} \mu_\O- \beta \mu_\O $ 
\item $|V^\AAB| \propto  (1 - \pbar )\mu_\A \mu_\text{AB} - \pbar \mu_\text{AB} \mu_\O -\beta \left( \mu_\A + \mu_\O\right)$
\end{enumerate}

\item Remaining $\O$-type NDDs donate to remaining under-demanded vertices. The previous step exhausted $N^\O$, so none of these donations occur. 

\item 2-cycles are created with $V^\ABO$ and remaining under-demanded vertices. The previous step exhausted $V^\ABO$, so none of these cycles occur.

\end{enumerate}

In the efficient matching described above, the number of \textit{matched} pairs in each under-demanded group is 

\begin{enumerate}[label={}]
\item $|V^\OA| \propto  \mu_\O \left(\beta +\pbar \left(\mu_\A+\mu_\text{AB}\right)\right) $ 
\item $|V^\OB| \propto \pbar \mu_\B \mu_\O$ 
\item $|V^\AAB| \propto (\beta + \pbar \mu_\text{AB}) (\mu_\A + \mu_\O)$ 
\item $|V^\BAB| \propto (\beta+ \pbar \mAB) \mu_\B$ 
\item $|V^\OAB| = 0 $ 
\end{enumerate}

Combining these with the over-demanded and self-demanded vertices, the total size of the efficient matching is

\begin{align*}
u_E &=  \pbar \Big[ 2\mAB \mu_\B +  2\mAB \mu_\A +  3\mAB \mu_\O \\
 +& 2\mu_\A\mu_\O +  2\mu_\B \mu_\O + \mu_\O^2 + \mu_\A^2 + \mu_\B^2 +\mAB^2 \Big] \\
+& 2\mu_\A \mu_\B + \beta \left(\mu_\A + \mu_\B + 2\mu_\O \right) 
\end{align*}

This efficient matching includes all highly sensitized vertices except for those in $V_H^\OAB$. To calculate the price of fairness we now find the size of the fair matching. We match each vertex in $V_H^\OAB$ by removing a 3-cycle of the form ($\ABO$, $\OA$, $\AAB$) and creating a 2-cycle ($\ABO$, $\OAB$). This matching used $|V^\ABO|\propto \pbar \mu_\O \mAB$ 3-cycles of this form, while $|V_H^\OAB| \propto (1-\lambda) \mu_\O \mAB$. The model assumption $\pbar > 1-\lambda$ ensures that $|V^\ABO|>|V_H^\OAB| $, and all vertices in $V_H^\OAB$ can be matched in this way.

To match each vertex in $V_H^\OAB$, we remove from the matching one vertex from both $V^\OA$ and $V^\AAB$. Thus the total efficiency loss is $|V_H^\OAB|\propto (1-\lambda) \mu_\O \mAB$. The price of fairness is

\begin{align*}
\POF{\mathcal{M}}{u_{LEX}} &= \frac{(1-\lambda)\mu_\O \mAB}{ u_E }
\end{align*}

With $u_E$ defined previously.
\end{proof}

\begin{proposition}\label{prop:midbeta11}
Assume
\begin{enumerate}[label=\textbf{\arabic*}]
\item $\beta < \mAB (1-\pbar)-\mAB\mu_\O\pbar/(\mu_\A+\mu_\B)$\label{midb1:b1}
\item  $\beta <  \frac{\mu_\A \mAB(1-\pbar )+\mu_\B \mu_\O(1-\pbar)-\pbar \mu_\O \mAB }{\mu_\A+\mu_\O}$ \label{midb1:b2}
\item $ \beta >\mAB (1-\pbar)-\mAB \mu_\O \pbar/\mu_\A$ \label{midb1:b3}
\item $\beta <\mAB(1-\pbar )-\pbar \mAB\mu_\O/\mu_\A + (1-\pbar) \mu_\B \mu_\O /\mu_A$ \label{midb1:b4}
\item $\beta < \mAB (1-\pbar)-\mu_\O \mAB/(1-\mAB)$ \label{midb1:b5}
\end{enumerate}

Note that as written, constraint \ref{midb1:b4} is a looser bound than \ref{midb1:b5}, and can be removed. However it is convenient to leave \ref{midb1:b4} for clarity. These constraints imply $\beta \in [0,1/12]$. Denote by $\Mcal{}$ the set of matchings in $G(n)$ using cycles and chains up to length 3. Almost surely as $n \rightarrow \infty$, the price of fairness is

\begin{align*}
\POF{\mathcal{M}}{u_{LEX}} &= \frac{(1-\lambda)\mu_\O \mAB}{ u_E  }
\end{align*}

with 

\begin{align*}
u_E &=  \pbar \Big[ 2\mAB \mu_\B +  2\mAB \mu_\A +  3\mAB \mu_\O \\
 +& 2\mu_\A\mu_\O +  2\mu_\B \mu_\O + \mu_\O^2 + \mu_\A^2 + \mu_\B^2+ \mAB^2 \Big] \\
+& 2\mu_\A \mu_\B + \beta \left(\mu_\A + \mu_\B + 2\mu_\O \right) 
\end{align*}
\end{proposition}

\begin{proof}[sketch]

We begin with matching $M^*$ as done in the proof of Lemma \ref{lem:underdemand}, matching all highly sensitized vertices except for those in $V_H^\ABO$. We now complete the efficient matching using both 3-cycles and 3-chains as in~\cite{Dickerson12:Optimizing}.

\begin{enumerate}\setcounter{enumi}{6}

\item $\A$- and $\B$-type NDDs donate to $V^\AAB$ and $V^\BAB$, respectively. Note that $|N^A|\propto \beta \mu_A$ and $|V^\AAB|\propto  (1-\pbar)\mu_A\mAB$. The inequality $\beta \mu_\A < \mu_A\mAB(1-\pbar)$ holds due to assumption \ref{smallb:b1},  and a.s. $|N^\A|<|V^\AAB|$. By the same argument, a.s. $|N^\B|<|V^\BAB|$. Thus, both $N^\A$ and $N^\B$ are exhausted, and $|V^\AAB| \propto \mu_\A \mAB (1-\pbar) - \beta \mu_\A $ and $|V^\BAB| \propto \mu_\B \mAB (1-\pbar) - \beta \mu_\B $.

\item Create cycles of the form ($\ABO$, $\OA$, $\AAB$). The current size of each vertex group is

\begin{enumerate}[label={(\arabic*)}]
\item $|V^\ABO| \propto \pbar \mAB \mu_\O$
\item $|V^\OA| \propto (1-\pbar )\mu_\A \mu_\O$
\item $|V^\AAB| \propto \mu_\A \mAB (1-\pbar ) - \beta \mu_\A$
\end{enumerate}

Note that the inequality $(3)<(1)$ can be written as

$$ \beta > \mAB (1 - \pbar) - \pbar \mAB \mu_\O/\mu_\A$$

which holds by assumption \ref{midb1:b3}, and a.s.  $|V^\AAB| < |V^\ABO|$. The inequality $(3)<(2)$ can be written as

$$\beta > (1 - \pbar)(\mAB- \mu_\O) $$

which holds by model assumptions, and a.s. $|V^\AAB| < |V^\OA|$. Executing these cycles exhausts $V^\AAB$ and leaves the following vertices remaining

\begin{enumerate}[label={}]
\item $|V^\OA| \propto (1-\pbar) \mu_\A (\mu_\O -\mAB) + \mu_\A \beta $
\item $|V^\ABO| \propto  \pbar \mAB \mu_\O - \mu_\A \mAB(1-\pbar)  + \mu_\A \beta$
\end{enumerate}

\item Create cycles of the form ($\ABO$, $\OB$, $\BAB$). The current size of each vertex group is

\begin{enumerate}[label={(\arabic*)}]
\item $|V^\ABO| \propto \pbar \mAB \mu_\O - \mu_\A \mAB(1-\pbar)  + \mu_\A \beta$
\item $|V^\OB| \propto (1-\pbar )\mu_\B \mu_\O$
\item $|V^\BAB| \propto \mu_\B \mAB (1-\pbar ) - \beta \mu_\B$
\end{enumerate}

Inequality $(1)<(2)$ can be written as

$$\beta <\mAB(1-\pbar )-\pbar \mAB\mu_\O/\mu_\A + (1-\pbar) \mu_\B \mu_\O /\mu_A $$

which holds by assumption \ref{midb1:b4}. Inequality $(1)<(3)$ can be written as

$$\beta < \mAB (1-\pbar)-\mAB\mu_\O\pbar/(\mu_\A+\mu_\B)$$

which holds by assumption \ref{midb1:b1}.

Executing these cycles exhausts $V^\ABO$ and leaves the following vertices remaining

\begin{enumerate}[label={}]
\item $|V^\OB| \propto \mu_\A\mAB (1-\pbar )+(\mu_\B-\pbar (\mAB+\mu_\B)) \mu_\O-\beta \mu_\A $
\item $|V^\BAB| \propto \left((1-\pbar) \mAB-\beta\right) \left(\mu_\A+\mu_\B\right)-\pbar \mAB \mu_\O $
\end{enumerate}

\item Create chains of the form  ($\O$,$\OA$,$\AAB$). Previous steps exhausted $V^\AAB$ so none of these chains occur.

\item Create chains of the form  ($\O$,$\OB$,$\BAB$). The current size of each vertex group is

\begin{enumerate}[label={(\arabic*)}]
\item $|N^\O| \propto \beta\mu_\O$
\item $|V^\OB| \propto \mu_\A\mAB (1-\pbar )+(\mu_\B-\pbar (\mAB+\mu_\B)) \mu_\O-\beta \mu_\A $
\item $|V^\BAB| \propto \left((1-\pbar) \mAB-\beta\right) \left(\mu_\A+\mu_\B\right)-\pbar \mAB \mu_\O $
\end{enumerate}

The inequality $(1)<(2)$ can be written as 

$$\beta <  \frac{\mu_\A \mAB(1-\pbar )+\mu_\B \mu_\O(1-\pbar)-\pbar \mu_\O \mAB }{\mu_\A+\mu_\O}$$

which holds by assumption \ref{midb1:b2}. The inequality $(1)<(3)$ can be written as 

$$\beta < \mAB (1-\pbar)-\mu_\O \mAB/(1-\mAB)$$

which holds by assumption \ref{midb1:b5}. Executing these chains exhausts $N^\O$ and leaves the following vertices remaining

\begin{enumerate}[label={}]
\item \begin{align*}|V^\OB| &\propto  \mu_\A\mAB (1-\pbar )+(\mu_\B-\pbar (\mAB+\mu_\B)) \mu_\O \\&-\beta (\mu_\A+\mu_\O) \end{align*}
\item $|V^\BAB| \propto (1-\pbar) \mAB-\beta) (\mu_\A+\mu_\B)-(\beta+\pbar \mAB) \mu_\O$
\end{enumerate}

\item Remaining $\O$-type NDDs donate to remaining under-demanded vertices. The previous step exhausted $N^\O$, so none of these donations occur. 

\item 2-cycles are created with $V^\ABO$ and remaining under-demanded vertices. The previous steps exhausted $V^\ABO$, so none of these cycles occur.

\end{enumerate}

In the efficient matching described above, the number of \textit{matched} pairs in each under-demanded group is 

\begin{enumerate}[label={}]
\item $|V^\OA| \propto (1-\pbar) \mu_\A (\mu_\O -\mAB) + \mu_\A \beta $
\item \begin{align*}|V^\OB| &\propto  \mu_\A\mAB (1-\pbar )+(\mu_\B-\pbar (\mAB+\mu_\B)) \mu_\O \\&-\beta (\mu_\A+\mu_\O)\end{align*}
\item $|V^\AAB| = 0$ 
\item $|V^\BAB| \propto (1-\pbar) \mAB-\beta) (\mu_\A+\mu_\B)-(\beta+\pbar \mAB) \mu_\O$
\item $|V^\OAB| = 0 $
\end{enumerate}

Combining these with the over-demanded and self-demanded vertices, the total size of the efficient matching is

\begin{align*}
u_E &=  \pbar \Big[ 2\mAB \mu_\B +  2\mAB \mu_\A +  3\mAB \mu_\O \\
 +& 2\mu_\A\mu_\O +  2\mu_\B \mu_\O + \mu_\O^2 + \mu_\A^2 + \mu_\B^2+ \mAB^2 \Big] \\
+& 2\mu_\A \mu_\B + \beta \left(\mu_\A + \mu_\B + 2\mu_\O \right) 
\end{align*}

This efficient matching includes all highly sensitized vertices except for those in $V_H^\OAB$. To calculate the price of fairness we now find the size of the fair matching. We match each vertex in $V_H^\OAB$ by removing a 3-cycle of the form ($\ABO$, $\OA$, $\AAB$) and creating a 2-cycle ($\ABO$, $\OAB$). This matching used $|V^\ABO|\propto \pbar \mu_\O \mAB$ 3-cycles of this form, while $|V_H^\OAB| \propto (1-\lambda) \mu_\O \mAB$. The model assumptions ensure that $|V^\ABO|>|V_H^\OAB| $, and all vertices in $V_H^\OAB$ can be matched in this way.

To match each vertex in $V_H^\OAB$, we remove from the matching one vertex from both $V^\OA$ and $V^\AAB$. Thus the total efficiency loss is $|V_H^\OAB|\propto (1-\lambda) \mu_\O \mAB$. The price of fairness is

\begin{align*}
\POF{\mathcal{M}}{u_{LEX}} &= \frac{(1-\lambda)\mu_\O \mAB}{ u_E }
\end{align*}

With $u_E$ defined previously.
\end{proof}

\begin{proposition}\label{prop:midbeta13}
Assume
\begin{enumerate}[label=\textbf{\arabic*}]
\item $ \beta >\mAB (1-\pbar)-\mAB\mu_\O\pbar/\mu_\A$ \label{midb13:b1}
\item $\beta < \mAB (1-\pbar)-\mAB\mu_\O\pbar/(\mu_\A+\mu_\B)$ \label{midb13:b2}
\item $\beta <\mAB(1-\pbar )-\pbar \mAB\mu_\O/\mu_\A + (1-\pbar) \mu_\B \mu_\O /\mu_A $ \label{midb13:b4}
\item $\beta >  \mAB\left( (1-\pbar)-\frac{\mu_\O}{1-\mAB}\right)$ \label{midb13:b5}
\item $\beta<\mAB (1-\pbar)-\lambda \mu_\O\frac{\mAB}{1-\mAB}$ \label{midb13:b6}
\end{enumerate}

These constraints imply $\beta \in [0,1/8]$. Denote by $\Mcal{}$ the set of matchings in $G(n)$ using cycles and chains up to length 3. Almost surely as $n \rightarrow \infty$, the price of fairness is

\begin{align*}
\POF{\mathcal{M}}{u_{LEX}}&= \frac{(1-\mAB) ((1-\pbar) \mAB-\beta)-\lambda \mAB \mu_\O}{u_E}
\end{align*}

with 

\begin{align*}
u_E &=  \mAB \mu_\B + \mu_\A \left(\mAB + 2 \mu_\B\right) + \beta \mu_\O \\
  &+\pbar \big[\mu_\A^2 + \mu_\A \mAB + \mAB^2 + \mAB \mu_\B + \mu_\B^2 \\
  &+ 2 (\mu_\A + \mAB + \mu_\B) \mu_\O + \mu_\O^2\big]
\end{align*}
\end{proposition}

\begin{proof}[sketch]

We begin with matching $M^*$ as done in the proof of Lemma \ref{lem:underdemand}, matching all highly sensitized vertices except for those in $V_H^\ABO$. We now complete the efficient matching using both 3-cycles and 3-chains as in~\cite{Dickerson12:Optimizing}. 

\begin{enumerate}\setcounter{enumi}{6}

\item $\A$- and $\B$-type NDDs donate to $V^\AAB$ and $V^\BAB$, respectively. Note that $|N^A|\propto \beta \mu_A$ and $|V^\AAB|\propto  (1-\pbar)\mu_A\mAB$. The inequality $\beta \mu_\A < \mu_A\mAB(1-\pbar)$ holds due to assumption \ref{midb13:b2},  and a.s. $|N^\A|<|V^\AAB|$. By the same argument, a.s. $|N^\B|<|V^\BAB|$. Thus, both $N^\A$ and $N^\B$ are exhausted, and $|V^\AAB| \propto \mu_\A \mAB (1-\pbar) - \beta \mu_\A $ and $|V^\BAB| \propto \mu_\B \mAB (1-\pbar) - \beta \mu_\B $.

\item Create cycles of the form ($\ABO$, $\OA$, $\AAB$). The current size of each vertex group is

\begin{enumerate}[label={(\arabic*)}]
\item $|V^\ABO| \propto \pbar \mAB \mu_\O$
\item $|V^\OA| \propto (1-\pbar )\mu_\A \mu_\O$
\item $|V^\AAB| \propto \mu_\A \mAB (1-\pbar ) - \beta \mu_\A$
\end{enumerate}

Note that the inequality $(3)<(1)$ can be written as

$$ \beta > \mAB (1 - \pbar) - \pbar \mAB \mu_\O/\mu_\A$$

which holds by assumption \ref{midb13:b1}, and a.s.  $|V^\AAB| < |V^\ABO|$. The inequality $(3)<(2)$ can be written as

$$\beta > (1 - \pbar)(\mAB - \mu_\O) $$

which holds by model assumptions, and a.s. $|V^\AAB| < |V^\OA|$. Executing these cycles exhausts $V^\AAB$ and leaves the following vertices remaining

\begin{enumerate}[label={}]
\item $|V^\OA| \propto (1-\pbar) \mu_\A (\mu_\O -\mAB) + \mu_\A \beta $
\item $|V^\ABO| \propto  \pbar \mAB \mu_\O - \mu_\A \mAB(1-\pbar)  + \mu_\A \beta$
\end{enumerate}

\item Create cycles of the form ($\ABO$, $\OB$, $\BAB$). The current size of each vertex group is

\begin{enumerate}[label={(\arabic*)}]
\item $|V^\ABO| \propto \pbar \mAB \mu_\O - \mu_\A \mAB(1-\pbar)  + \mu_\A \beta$
\item $|V^\OB| \propto (1-\pbar )\mu_\B \mu_\O$
\item $|V^\BAB| \propto \mu_\B \mAB (1-\pbar ) - \beta \mu_\B$
\end{enumerate}

Inequality $(1)<(2)$ can be written as

$$\beta <\mAB(1-\pbar )-\pbar \mAB\mu_\O/\mu_\A + (1-\pbar) \mu_\B \mu_\O /\mu_A $$

which holds by assumption \ref{midb13:b4}. Inequality $(1)<(3)$ can be written as

$$\beta < \mAB (1-\pbar)-\mAB\mu_\O\pbar/(\mu_\A+\mu_\B)$$

which holds by assumption \ref{midb13:b2}.

Executing these cycles exhausts $V^\ABO$ and leaves the following vertices remaining

\begin{enumerate}[label={}]
\item $|V^\OB| \propto \mu_\A\mAB (1-\pbar )+(\mu_\B-\pbar (\mAB+\mu_\B)) \mu_\O-\beta \mu_\A $
\item $|V^\BAB| \propto \left((1-\pbar) \mAB-\beta\right) \left(\mu_\A+\mu_\B\right)-\pbar \mAB \mu_\O $
\end{enumerate}

\item Create chains of the form  ($\O$,$\OA$,$\AAB$). Previous steps exhausted $V^\AAB$ so none of these chains occur.

\item Create chains of the form  ($\O$,$\OB$,$\BAB$). The current size of each vertex group is

\begin{enumerate}[label={(\arabic*)}]
\item $|N^\O| \propto \beta\mu_\O$
\item $|V^\OB| \propto \mu_\A\mAB (1-\pbar )+(\mu_\B-\pbar (\mAB+\mu_\B)) \mu_\O-\beta \mu_\A $
\item $|V^\BAB| \propto \left((1-\pbar) \mAB-\beta\right) \left(\mu_\A+\mu_\B\right)-\pbar \mAB \mu_\O $
\end{enumerate}

The inequality $(3)<(1)$ can be written as 

$$\beta > \mAB\left( (1-\pbar)-\frac{\mu_\O}{1-\mAB}\right)$$

which holds by assumption \ref{midb13:b5}. The inequality $(3)<(2)$ can be written as 

$$\beta > (1-\pbar)(\mAB - \mu_\O)$$

which holds by the model assumptions. Executing these chains exhausts $V^\BAB$ and leaves the following vertices remaining

\begin{enumerate}[label={}]
\item $|N^\O| \propto (\beta + \pbar \mAB) (\mu_\A + \mu_\B + \mu_\O) - \mAB (\mu_\A + \mu_\B)$
\item $|V^\OB| \propto \mu_\B \left((\beta+(1-\pbar) (\mu_\O-\mAB)\right)$
\end{enumerate}

\item \label{ostep} Remaining $\O$-type NDDs and $V^\ABO$ vertices match with remaining under-demanded vertices, starting with $V^\OAB$. The remaining size of each vertex group is

\begin{enumerate}[label={(\arabic*)}]
\item $|N^\O| \propto (\beta + 
    \pbar \mAB) (\mu_\A + \mu_\B + \mu_\O) - \mAB (\mu_\A + \mu_\B)$
\item $|V^\OAB| \propto \mAB \mu_\O $
\item $|V^\ABO| = 0 $
\end{enumerate}

After simplifying, the inequality $(1)<(2)$ can be written as

$$\beta < \mAB(1-\pbar) $$

which holds by assumption \ref{midb13:b2}. Thus $\O$-type NDDs are exhausted first, leaving some vertices remaining in $V^\OAB$, with 

$$ |V^\OAB| \propto \left((1-\pbar) \mAB-\beta\right) (1-\mAB)$$

\end{enumerate}

In the efficient matching described above, the number of \textit{matched} pairs in each under-demanded group is 

\begin{enumerate}[label={}]
\item $|V^\OA| \propto   \mu_\A \left(\mAB(1-\pbar)+\pbar \mu_\O-\beta\right)$
\item $|V^\OB| \propto  \mu_\B \left(\mAB(1-\pbar)+\pbar \mu_\O-\beta\right) $
\item $|V^\AAB| \propto \mu_\A \mAB $ 
\item $|V^\BAB| \propto \mu_\B \mAB$
\item $|V^\OAB| \propto (\beta +\mAB \pbar)(1-\mAB) -\mAB (\mu_\A+\mu_\B)$ 
\end{enumerate}

Combining these with the over-demanded and self-demanded vertices, the total size of the efficient matching is

\begin{align*}
u_E &=  \mAB \mu_\B + \mu_\A \left(\mAB + 2 \mu_\B\right) + \beta \mu_\O \\
  &+\pbar \big[\mu_\A^2 + \mu_\A \mAB + \mAB^2 + \mAB \mu_\B + \mu_\B^2 \\
  &+ 2 (\mu_\A + \mAB + \mu_\B) \mu_\O + \mu_\O^2\big]
\end{align*}

To calculate the price of fairness we now find the size of the fair matching. The only unmatched highly sensitized patients are in $V_H^\OAB$, some of which were matched in step \ref{ostep} above. We now show that the number of matched vertices in $V^\OAB$ is smaller than the initial size of $V_H^\OAB$, so not all vertices in $V_H^\OAB$ can be matched. Let $M^\OAB$ be the number of matched vertices in $V^\OAB$, and $H^\OAB$ be the initial size of $V_H^\OAB$. The inequality $M^\OAB<H^\OAB$ can be written as

\begin{align}
(\beta +\mAB \pbar)(1&-\mAB) -\mAB (\mu_\A+\mu_\B) < (1-\lambda)\mu_\O\mAB \\
\beta &< \mAB (1-\pbar)-\lambda \mu_\O\frac{\mAB}{1-\mAB}
\end{align}

This inequality holds by assumption \ref{midb13:b6}, and a.s. there are some unmatched vertices in $V_H^\OAB$. The number of unmatched highly sensitized vertices is 

$$H^\OAB-M^\OAB \propto (1-\mAB) ((1-\pbar) \mAB-\beta)-\lambda \mAB \mu_\O$$.

We match each of these remaining vertices by removing a 3-cycle of the form ($\ABO$, $\OA$, $\AAB$) and creating a 2-cycle ($\ABO$, $\OAB$). This matching used $|V^\ABO|\propto \pbar \mu_\O \mAB$ 3-cycles of this form, while $|V_H^\OAB| \propto (1-\lambda) \mu_\O \mAB$. The model assumptions ensure that $|V^\ABO|>|V_H^\OAB| $, and all remaining vertices in $V_H^\OAB$ can be matched in this way.

To match each remaining vertex in $V_H^\OAB$, we remove from the matching one vertex from both $V^\OA$ and $V^\AAB$. Thus the total efficiency loss is $H^\OAB-M^\OAB$. The price of fairness is

\begin{align*}
\POF{\mathcal{M}}{u_{LEX}} &= \frac{(1-\mAB) ((1-\pbar) \mAB-\beta)-\lambda \mAB \mu_\O}{ u_E }
\end{align*}

With $u_E$ defined previously.
\end{proof}

\begin{proposition}\label{prop:midbeta3}
Assume
\begin{enumerate}[label=\textbf{\arabic*}]
\item $\beta > \mAB (1-\pbar)-\mAB\mu_\O\pbar/(\mu_\A+\mu_\B)$\label{midb3:b1}
\item $\beta<\mAB (1-\pbar)-\lambda \mu_\O\frac{\mAB}{1-\mAB}$ \label{midb3:b4}
\end{enumerate}

These constraints imply $\beta \in [0,1/10]$. Denote by $\Mcal{}$ the set of matchings in $G(n)$ using cycles and chains up to length 3. Almost surely as $n \rightarrow \infty$, the price of fairness is

\begin{align*}
\POF{\mathcal{M}}{u_{LEX}} &= \frac{(1-\mAB) ((1-\pbar) \mAB-\beta)-\lambda \mAB \mu_\O}{u_E}
\end{align*}

with 

\begin{align*}
u_E &=  \mAB \mu_\B + \mu_\A \left(\mAB + 2 \mu_\B\right) + \beta \mu_\O \\
  &+\pbar \big[\mu_\A^2 + \mu_\A \mAB + \mAB^2 + \mAB \mu_\B + \mu_\B^2 \\
  &+ 2 (\mu_\A + \mAB + \mu_\B) \mu_\O + \mu_\O^2\big]
\end{align*}
\end{proposition}

\begin{proof}[sketch]

We begin with matching $M^*$ as done in the proof of Lemma \ref{lem:underdemand}, matching all highly sensitized vertices except for those in $V_H^\ABO$. We now complete the efficient matching using both 3-cycles and 3-chains as in~\cite{Dickerson12:Optimizing}. 

\begin{enumerate}\setcounter{enumi}{6}

\item $\A$- and $\B$-type NDDs donate to $V^\AAB$ and $V^\BAB$, respectively. Note that $|N^A|\propto \beta \mu_A$ and $|V^\AAB|\propto  (1-\pbar)\mu_A\mAB$. The inequality $\beta \mu_\A < \mu_A\mAB(1-\pbar)$ holds due to assumption \ref{midb3:b4},  and a.s. $|N^\A|<|V^\AAB|$. By the same argument, a.s. $|N^\B|<|V^\BAB|$. Thus, both $N^\A$ and $N^\B$ are exhausted, and $|V^\AAB| \propto \mu_\A \mAB (1-\pbar) - \beta \mu_\A $ and $|V^\BAB| \propto \mu_\B \mAB (1-\pbar) - \beta \mu_\B $.

\item Create cycles of the form ($\ABO$, $\OA$, $\AAB$). The current size of each vertex group is

\begin{enumerate}[label={(\arabic*)}]
\item $|V^\ABO| \propto \pbar \mAB \mu_\O$
\item $|V^\OA| \propto (1-\pbar )\mu_\A \mu_\O$
\item $|V^\AAB| \propto \mu_\A \mAB (1-\pbar ) - \beta \mu_\A$
\end{enumerate}

Note that the inequality $(3)<(1)$ can be written as

$$ \beta > \mAB (1 - \pbar) - \pbar \mAB \mu_\O/\mu_\A$$

which holds by assumption \ref{midb3:b1} and a.s.  $|V^\AAB| < |V^\ABO|$. The inequality $(3)<(2)$ can be written as

$$\beta > (1 - \pbar)(\mAB - \mu_\O) $$

which holds by the model assumptions, and a.s. $|V^\AAB| < |V^\OA|$. Executing these cycles exhausts $V^\AAB$ and leaves the following vertices remaining

\begin{enumerate}[label={}]
\item $|V^\OA| \propto (1-\pbar) \mu_\A (\mu_\O -\mAB) + \mu_\A \beta $
\item $|V^\ABO| \propto  \pbar \mAB \mu_\O - \mu_\A \mAB(1-\pbar)  + \mu_\A \beta$
\end{enumerate}

\item Create cycles of the form ($\ABO$, $\OB$, $\BAB$). The current size of each vertex group is

\begin{enumerate}[label={(\arabic*)}]
\item $|V^\ABO| \propto \pbar \mAB \mu_\O - \mu_\A \mAB(1-\pbar)  + \mu_\A \beta$
\item $|V^\OB| \propto (1-\pbar )\mu_\B \mu_\O$
\item $|V^\BAB| \propto \mu_\B \mAB (1-\pbar ) - \beta \mu_\B$
\end{enumerate}

Inequality $(3)<(2)$ can be written as

$$\beta > (1-\pbar)(\mAB-\mu_\O )$$

which holds by the model assumptions. Inequality $(3)<(1)$ can be written as

$$\beta > \mAB (1-\pbar)-\mAB\mu_\O\pbar/(\mu_\A+\mu_\B)$$

which holds by assumption \ref{midb3:b1}.

Executing these cycles exhausts $V^\BAB$ and leaves the following vertices remaining

\begin{enumerate}[label={}]
\item $|V^\ABO| \propto (\beta - (1 - \pbar) \mAB) (\mu_\A + \mu_\B) + 
 \pbar \mAB \mu_\O $
\item $|V^\BAB| \propto \mu_\B (\beta - (1 - \pbar) (\mAB - \mu_\O)) $
\end{enumerate}

\item Create chains of the form  ($\O$,$\OA$,$\AAB$). Previous steps exhausted $V^\AAB$ so none of these chains occur.

\item Create chains of the form  ($\O$,$\OB$,$\BAB$). Previous steps exhausted $V^\BAB$ so none of these chains occur.

\item $\O$-type NDDs and $V^\ABO$ match with remaining under-demanded vertices, starting with $V^\OAB$. The remaining size of each vertex group is 

\begin{enumerate}[label={(\arabic*)}]
\item $|N^\O| \propto \beta \mu_\O$
\item $|V^\ABO| \propto (\beta - (1 - \pbar) \mAB) (\mu_\A + \mu_\B) + 
 \pbar \mAB \mu_\O$
\item $|V^\OAB| \propto \mu_\O \mAB $
\end{enumerate}

Note that the inequality $(1)+(2)<(3)$ can be written as

$$\beta < \mAB(1-\pbar) $$

which holds by assumption \ref{midb3:b4} Thus $\O$-type NDDs are exhausted first, leaving some vertices remaining in $V^\OAB$, with 

$$ |V^\OAB| \propto \left((1-\pbar) \mAB-\beta) (1-\mAB\right)$$

\end{enumerate}

In the efficient matching described above, the number of \textit{matched} pairs in each under-demanded group is 

\begin{enumerate}[label={}]
\item $|V^\OA| \propto   \mu_\A (\mAB+\pbar( \mu_\O-\mAB)-\beta)$
\item $|V^\OB| \propto  \mu_\B (\mAB+\pbar( \mu_\O-\mAB)-\beta) $
\item $|V^\AAB| \propto \mu_\A \mAB $ 
\item $|V^\BAB| \propto \mu_\B \mAB$.
\item $|V^\OAB| \propto (\beta +\mAB \pbar)(1-\mAB) -\mAB (\mu_\A+\mu_\B) $ 
\end{enumerate}

Combining these with the over-demanded and self-demanded vertices, the total size of the efficient matching is

\begin{align*}
u_E &=  \mAB \mu_\B + \mu_\A \left(\mAB + 2 \mu_\B\right) + \beta \mu_\O \\
  &+\pbar \big[\mu_\A^2 + \mu_\A \mAB + \mAB^2 + \mAB \mu_\B + \mu_\B^2 \\
  &+ 2 (\mu_\A + \mAB + \mu_\B) \mu_\O + \mu_\O^2\big]
\end{align*}

To calculate the price of fairness we now find the size of the fair matching. The only unmatched highly sensitized patients are in $V_H^\OAB$, some of which were matched in step \ref{ostep} above. We now show that the number of matched vertices in $V^\OAB$ is smaller than the initial size of $V_H^\OAB$, so not all vertices in $V_H^\OAB$ can be matched. Let $M^\OAB$ be the number of matched vertices in $V^\OAB$, and $H^\OAB$ be the initial size of $V_H^\OAB$. The inequality $M^\OAB<H^\OAB$ can be written as

\begin{align}
(\beta +\mAB \pbar)(1&-\mAB) -\mAB (\mu_\A+\mu_\B)  < (1-\lambda)\mu_\O\mAB \\
\beta &< \mAB (1-\pbar)-\lambda \mu_\O\frac{\mAB}{1-\mAB}
\end{align}

This inequality holds by assumption \ref{midb3:b4}, and a.s. there are some unmatched vertices in $V_H^\OAB$. The number of unmatched highly sensitized vertices is 

$$H^\OAB-M^\OAB \propto (1-\mAB) ((1-\pbar) \mAB-\beta)-\lambda \mAB \mu_\O.$$

We match each of these remaining vertices by removing a 3-cycle of the form ($\ABO$, $\OA$, $\AAB$) and creating a 2-cycle ($\ABO$, $\OAB$). This matching used $|V^\ABO|\propto \pbar \mu_\O \mAB$ 3-cycles of this form, while $|V_H^\OAB| \propto (1-\lambda) \mu_\O \mAB$. The model assumptions ensure that $|V^\ABO|>|V_H^\OAB| $, and all remaining vertices in $V_H^\OAB$ can be matched in this way.

To match each remaining vertex in $V_H^\OAB$, we remove from the matching one vertex from both $V^\OA$ and $V^\AAB$. Thus the total efficiency loss is $H^\OAB-M^\OAB$. The price of fairness is

\begin{align*}
\POF{\mathcal{M}}{u_{LEX}} &= \frac{(1-\mAB) ((1-\pbar) \mAB-\beta)-\lambda \mAB \mu_\O}{ u_E }
\end{align*}

With $u_E$ defined previously.
\end{proof}

Next we compare the price of fairness in Propositions \ref{prop:smallbeta}, \ref{prop:midbeta11}, \ref{prop:midbeta13}, and \ref{prop:midbeta3} to the price of fairness in the efficient matching without NDDs, given in~\englishcite{Dickerson14:Price}:

\begin{align}\label{eq:dickersonpof}
\text{POF}_0 &= \frac{(1-\lambda)\mu_\O \mAB}{ u_E  }
\end{align}

\begin{align*}
u_E &=  \pbar \Big[ 2\mAB \mu_\B +  2\mAB \mu_\A +  3\mAB \mu_\O \\
 +& 2\mu_\A\mu_\O +  2\mu_\B \mu_\O + \mu_\O^2 + \mu_\A^2 + \mu_\B^2+ \mAB^2 \Big] \\
+& 2\mu_\A \mu_\B 
\end{align*}

The following Lemmas state that $\text{POF}_0$ is an upper bound on the price of fairness when NDDs are used, for each of the four cases when the price of fairness is nonzero.

\begin{lemma} \label{lem:prop12-bounded}
The price of fairness in Propositions \ref{prop:smallbeta} and \ref{prop:midbeta11} is bounded above by $\text{POF}_0$. 
\end{lemma}

\begin{proof}[sketch]
The price of fairness in Propositions \ref{prop:smallbeta} and \ref{prop:midbeta11} is

\begin{align*}
\text{POF}_A &= \frac{(1-\lambda)\mu_\O \mAB}{ u_E  }\\
\\
u_E &=  \pbar \Big[ 2\mAB \mu_\B +  2\mAB \mu_\A +  3\mAB \mu_\O \\
 +& 2\mu_\A\mu_\O +  2\mu_\B \mu_\O + \mu_\O^2 + \mu_\A^2 + \mu_\B^2+ \mAB^2 \Big] \\
+& 2\mu_\A \mu_\B + \beta \left(\mu_\A + \mu_\B + 2\mu_\O \right) 
\end{align*}

Both $\text{POF}_0$ and $\text{POF}_A$ have the same numerator, and the denominator of  $\text{POF}_A$ is equal to the denominator of $\text{POF}_0$, with the additional term $\beta \left(\mu_\A + \mu_\B + 2\mu_\O \right)$. Thus when $\beta=0$, $\text{POF}_0=\text{POF}_A$, and when $\beta>0$,  $\text{POF}_0>\text{POF}_A$, and the price of fairness in Propositions \ref{prop:smallbeta} and \ref{prop:midbeta11} is bounded above by $\text{POF}_0$.
\end{proof}

\begin{lemma} \label{lem:prop34-bounded}
The price of fairness in Propositions \ref{prop:midbeta13} and \ref{prop:midbeta3} is bounded above by $\text{POF}_0$. 
\end{lemma}

\begin{proof}[sketch]
The price of fairness in Propositions \ref{prop:midbeta13} and \ref{prop:midbeta3} is

\begin{align*}
\text{POF}_B &= \frac{(1-\mAB) ((1-\pbar) \mAB-\beta)-\lambda \mAB \mu_\O}{u_E}\\
\\
u_E &=  \mAB \mu_\B + \mu_\A \left(\mAB + 2 \mu_\B\right) + \beta \mu_\O \\
  &+\pbar \big[\mu_\A^2 + \mu_\A \mAB + \mAB^2 + \mAB \mu_\B + \mu_\B^2 \\
  &+ 2 (\mu_\A + \mAB + \mu_\B) \mu_\O + \mu_\O^2\big]
\end{align*}

To show that $\text{POF}_B<\text{POF}_0$ holds, we first show both (1) the numerator of $\text{POF}_B$ is smaller than that of $\text{POF}_0$, and (2) the denominator of $\text{POF}_B$ is larger than the denominator of $\text{POF}_0$.

\noindent \textbf{(1)} In both $\text{POF}_0$ and $\text{POF}_B$, the numerator is proportional to the number of remaining vertices in $V_H^\OAB$, after constructing the efficient matching. In Proposition \ref{prop:midbeta13} and \ref{prop:midbeta3} the efficient matching contains some vertices in $V_H^\OAB$; without NDDs, the efficient matching contains no vertices in $V_H^\OAB$. Thus, the numerator of $\text{POF}_B$ is strictly smaller the numerator of  $\text{POF}_0$.

\noindent \textbf{(2)} Let the $D_0$ be the denominator of $\text{POF}_0$, and $D_B$ be the denominator of $\text{POF}_B$. We now show that the inequality $D_0<D_B$ holds. First, note that this inequality can be written as

$$
\mAB -(1 - \pbar) \mAB^2 + \beta \mu_\O > \mAB (\pbar + \mu_\O). 
$$

Rearranging, we have

{\small
\begin{align}\label{eq:ineq2}
\beta &>(\mAB/\mu_\O)\big[ (1-\pbar)\mAB -(\mu_\A+\mu_\B+\mAB-\pbar)\big].
\end{align}
}

We now show that inequality \ref{eq:ineq2} is satisfied by the the following assumption on $\beta$, made in Propositions \ref{prop:midbeta13} and \ref{prop:midbeta3}:

\begin{equation*}\label{eq:ineq3}
\mathbf{A}:  \beta > \mAB (1-\pbar)-\mAB \mu_\O \pbar / \mu_\A.
 \end{equation*}

Next, we show that assumption $\mathbf{A}$ implies inequality \ref{eq:ineq2}, and thus assumption $\mathbf{A}$ implies $D_0<D_B$. Assumption $\mathbf{A}$ implies \ref{eq:ineq2} if the right-hand side of $\mathbf{A}$ is larger than the right hand side of \ref{eq:ineq2}, that is,

{\small
\begin{align*}
 \mAB (1-\pbar)-\mAB \mu_\O \pbar / \mu_\A  &>  (\mAB/\mu_\O)(1-\pbar)\mAB \\
 &-(\mAB/\mu_\O)(\mu_\A+\mu_\B+\mAB-\pbar)
\end{align*}
}

rearranging, we have

\begin{align*}
\frac{1-\pbar}{\pbar} > \frac{\mu_\O}{\mu_\A}\frac{1-\mAB-\mu_\B}{ 1 - \mu_\B }
\end{align*}

The random graph model assumes $\pbar\leq 2/5$, and $\mu_\O \leq (3/2)\mu_\A$, thus we have

\begin{align*}
\frac{1-\pbar}{\pbar} \geq \frac{3}{2}> \frac{3}{2} \frac{1-\mAB-\mu_\B}{ 1 - \mu_\B } \geq \frac{\mu_\O}{\mu_\A}\frac{1-\mAB-\mu_\B}{ 1 - \mu_\B }.
\end{align*}

This shows that assumption $\mathbf{A}$ implies $D_0<D_B$.

Thus, the numerator of $\text{POF}_0$ is larger than the numerator of $\text{POF}_B$, and the denominator of $\text{POF}_0$ is smaller than the denominator of $\text{POF}_B$, and therefore $\text{POF}_B<\text{POF}_0$.
\end{proof}

Lemmas \ref{lem:prop12-bounded} and \ref{lem:prop34-bounded} show that with $\beta>0$, the price of fairness has the same upper bound as when $\beta=0$, given in~\englishcite{Dickerson14:Price}. That is, adding NDDs to the random graph model does not increase the price of fairness. 

\begin{reptheorem}{thm:pofdecreases}
Adding NDDs to the random graph model ($\beta>0$) does not increase the upper bound on the price of fairness found by~\englishcite{Dickerson14:Price}.
\end{reptheorem}

\begin{proof}
When $\beta>0$, there are only four possible matchings with nonzero price of fairness, and the price of fairness for each case is given in Propositions \ref{prop:smallbeta}, \ref{prop:midbeta11}, \ref{prop:midbeta13}, and \ref{prop:midbeta3}. Lemmas \ref{lem:prop12-bounded} and \ref{lem:prop34-bounded} state that in each of these four cases, the matching with NDDs has a tighter bound on the price of fairness than the matching without NDDs, given in~\englishcite{Dickerson14:Price}.
\end{proof}

Next we show that the price of fairness is zero when $\beta>1/8$, by finding the maximum possible $\beta$ for each of the four cases with nonzero price of fairness.

\begin{lemma}\label{lem:maxbeta-small}
In the matching described by Proposition \ref{prop:smallbeta}, $\beta<1/8$.
\end{lemma}

\begin{proof}
Proposition \ref{prop:smallbeta} makes the following assumptions on $\beta$:

\begin{enumerate}[label=\textbf{\arabic*}]
\item  $\beta< \mu_\A \left(1-\pbar\right) - \pbar \mu_\text{AB}$ \label{lemsmallb:1}
\item $\beta < \mAB (1 - \pbar ) - \pbar \mAB \mu_\O / \mu_\A$ \label{lemsmallb:2}
\item $\beta < \mAB \left(\frac{\mu_\A}{\mu_\A+\mu_\O} - \pbar \right)$ \label{lemsmallb:3}
\end{enumerate}

To determine an upper bound on $\beta$, we maximize the right hand side of constraint \ref{lemsmallb:3}. Note that the model assumes $\mAB<1/4$, $\mu_\A<1/2$, and $\mu_\A+\mu_\O<1$. Using these bounds, and $\pbar \rightarrow 0$, constraint \ref{lemsmallb:3} is bounded by

\begin{align*}\beta &< \mAB \left(\frac{\mu_\A}{\mu_\A+\mu_\O} - \pbar \right) < (1/4) \frac{(1/2)}{1}=1/8 \\
\beta &< 1/8
\end{align*}

Constraints \ref{lemsmallb:1} and \ref{lemsmallb:2} are looser than constraint \ref{lemsmallb:3}: with the values $\pbar \rightarrow 0$, $\mu_\A \rightarrow1/4$, and $\mAB \rightarrow 1/4$, both constraints reduce to $\beta < 1/4$.

\end{proof}

\begin{lemma}\label{lem:maxbeta-mid11}
In the matching described by Proposition \ref{prop:midbeta11}, $\beta<1/12$.
\end{lemma}

\begin{proof}

Proposition \ref{prop:midbeta11} makes the following assumptions

\begin{enumerate}[label=\textbf{\arabic*}]
\item $\beta < \mAB (1-\pbar)-\mAB\mu_\O\pbar/(\mu_\A+\mu_\B)$\label{lem11:b1}
\item  $\beta <  \frac{\mu_\A \mAB(1-\pbar )+\mu_\B \mu_\O(1-\pbar)-\pbar \mu_\O \mAB }{\mu_\A+\mu_\O}$ \label{lem11:b2}
\item $ \beta >\mAB (1-\pbar)-\mAB \mu_\O \pbar/\mu_\A$ \label{lem11:b3}
\item $\beta <\mAB(1-\pbar )-\pbar \mAB\mu_\O/\mu_\A + (1-\pbar) \mu_\B \mu_\O /\mu_A$ \label{lem11:b4}
\item $\beta < \mAB (1-\pbar)-\mu_\O \mAB/(1-\mAB)$ \label{lem11:b5}
\end{enumerate}

Combining \ref{lem11:b3} and \ref{lem11:b5}, we have

\begin{align*}
\mu_\O \mAB/(1-\mAB) &<  \mAB (1-\pbar)-\beta <  \mAB \mu_\O \pbar/\mu_\A \\
\mu_\O \mAB/(1-\mAB) & <  \mAB \mu_\O \pbar/\mu_\A \\
\mathbf{A}: \mu_\A/(1-\mAB) &< \pbar\\
\end{align*}

Combining constraint $\mathbf{A}$ with \ref{lem11:b5} gives a new upper bound on $\beta$,

\begin{align*} 
\beta &< \mAB (1-\pbar)-\mu_\O \mAB/(1-\mAB) \\
&< \mAB (1-\mu_\A/(1-\mAB) )-\mu_\O \mAB/(1-\mAB)\\
\\
\beta &<  \mAB \left( 1-\frac{\mu_\A +\mu_\O}{1-\mAB}\right)
\end{align*}

This bound is maximized when when $\mAB$ is maximal, and $(\mu_\A+\mu_\O)$ is minimal. In the random graph model, these values are $\mAB\rightarrow 1/4$ and $(\mu_\A+\mu_\O)\rightarrow 1/2$, and the numerical bound is

\begin{align*}
\beta &<  (1/4) \left( 1-\frac{(1/2)}{1-1/4}\right)=1/12 \\
\beta &< 1/12
\end{align*}

\end{proof}

\begin{lemma}\label{lem:maxbeta-mid13}
In the matching described by Proposition \ref{prop:midbeta13},  $\beta<1/8$.
\end{lemma}

\begin{proof}

Proposition \ref{prop:midbeta13} makes the following assumptions on $\beta$

\begin{enumerate}[label=\textbf{\arabic*}]
\item $ \beta >\mAB (1-\pbar)-\mAB\mu_\O\pbar/\mu_\A$ \label{lem13:b1}
\item $\beta < \mAB (1-\pbar)-\mAB\mu_\O\pbar/(\mu_\A+\mu_\B)$ \label{lem13:b2}
\item $\beta <\mAB(1-\pbar )-\pbar \mAB\mu_\O/\mu_\A + (1-\pbar) \mu_\B \mu_\O /\mu_A $ \label{lem13:b4}
\item $\beta >  \mAB\left( (1-\pbar)-\frac{\mu_\O}{1-\mAB}\right)$ \label{lem13:b5}
\item $\beta<\mAB (1-\pbar)-\lambda \mu_\O\frac{\mAB}{1-\mAB}$ \label{lem13:b6}
\end{enumerate}

Combining \ref{lem13:b1} and \ref{lem13:b6} results in the following constraint, which is consistent with the above assumptions:

\begin{align*}
\mathbf{A} : \lambda \frac{ \mu_\A }{1-\mAB} &<  \pbar
\end{align*}

Note that \ref{lem13:b6} is maximized when $\lambda$ is minimized; this occurs when $\lambda+\pbar\rightarrow 1$, and $\lambda \rightarrow 1-\pbar$. In this case, \ref{lem13:b6} can be relaxed as 

\begin{align*}
\beta &< \mAB (1-\pbar)-\lambda \mu_\O\frac{\mAB}{1-\mAB} \\
 &< \mAB (1-\pbar)-(1-\pbar) \mAB \frac{\mu_\O}{1-\mAB} \\ 
 \\
\beta &< \mAB (1-\pbar)-(1-\pbar) \mAB \frac{\mu_\O}{1-\mAB} \\ 
 &= (1-\pbar) \frac{\mAB (\mu_\A+\mu_\B)}{1-\mAB}\\
 \end{align*}
 
 Finally, we have
 
 \begin{align*}
 \beta &< (1-\pbar) \frac{\mAB (\mu_\A+\mu_\B)}{1-\mAB}
\end{align*}

The right hand side of this constraint is maximal when $\pbar$ is minimal; constraint $\mathbf{A}$ determines the lower bound for $\pbar$, with $\lambda \rightarrow 1-\pbar$:

\begin{align*}
 (1-\pbar) \frac{ \mu_\A }{1-\mAB} &<  \pbar \\
 \frac{ \mu_\A }{1-\mAB} &<  \pbar\left( 1+ \frac{ \mu_\A }{1-\mAB}  \right) \\
\frac{\mu_\A}{1-\mAB+\mu_\A} &< \pbar 
\end{align*}

Using this lower bound on $\pbar$, we can further relax \ref{lem13:b6}

\begin{align*}
\beta &< (1-\pbar) \frac{\mAB (\mu_\A+\mu_\B)}{1-\mAB} \\
 &< (1-\frac{\mu_\A}{1-\mAB+\mu_\A}) \frac{\mAB (\mu_\A+\mu_\B)}{1-\mAB} \\
 &= \frac{1-\mAB}{1-\mAB+\mu_\A} \frac{\mAB (\mu_\A+\mu_\B)}{1-\mAB} \\
 &= \frac{\mAB (\mu_\A+\mu_\B)}{1-\mAB+\mu_\A} \\
 \\
 \beta &< \frac{\mAB (\mu_\A+\mu_\B)}{1-\mAB+\mu_\A}
\end{align*}

The right hand side is maximal when $\mAB$ is maximal, and $\mAB,\mu_\A,\mu_\B,\mu_\O \rightarrow 1/4$. This gives the final bound on $\beta$,

\begin{align*}
 \beta &< \frac{(1/4) (1/2)}{1} = 1/8\\
 \beta &< 1/8
\end{align*}

\end{proof}

\begin{lemma}\label{lem:maxbeta-mid3}
In the matching described by Proposition \ref{prop:midbeta3}, $\beta<1/10$.
\end{lemma}

\begin{proof}
Proposition \ref{prop:midbeta3} makes the following assumptions on $\beta$

\begin{enumerate}[label=\textbf{\arabic*}]
\item $\beta > \mAB (1-\pbar)-\mAB\mu_\O\pbar/(\mu_\A+\mu_\B)$\label{lem3:b1}
\item $\beta <\mAB (1-\pbar)-\lambda \mu_\O\frac{\mAB}{1-\mAB}$ \label{lem3:b4}
\end{enumerate}

Combining these assumptions results in the following constraint:

\begin{align*}
\mathbf{A} : \lambda \frac{ \mu_\A +\mu_\B}{1-\mAB} &<  \pbar
\end{align*}

Note that assumption \ref{lem3:b4} is identical to assumption \ref{lem13:b6} in Lemma \ref{lem:maxbeta-mid13}. Following the same procedure used in the proof of Lemma \ref{lem:maxbeta-mid13}, \ref{lem3:b4} can be relaxed as 

\begin{align*}	
 \beta &< (1-\pbar) \frac{\mAB (\mu_\A+\mu_\B)}{1-\mAB}
\end{align*}

The right hand side of this constraint is maximal when $\pbar$ is minimal; constraint $\mathbf{A}$ determines the lower bound for $\pbar$, with $\lambda \rightarrow 1-\pbar$:

\begin{align*}
 (1-\pbar) \frac{ \mu_\A +\mu_\B}{1-\mAB} &<  \pbar \\
 \frac{ \mu_\A+\mu_\B }{1-\mAB} &<  \pbar\left( 1+ \frac{ \mu_\A+\mu_\B }{1-\mAB}  \right) \\
\frac{\mu_\A+\mu_\B}{2\mu_\A+2\mu_\B+\mu_\O} &< \pbar 
\end{align*}

Using this lower bound on $\pbar$, we can further relax \ref{lem3:b4}

\begin{align*}
\beta &< (1-\pbar) \frac{\mAB (\mu_\A+\mu_\B)}{1-\mAB} \\
 &< (1-\frac{\mu_\A+\mu_\B}{2\mu_\A+2\mu_\B+\mu_\O}) \frac{\mAB (\mu_\A+\mu_\B)}{1-\mAB} \\
 &= \frac{1-\mAB}{2\mu_\A+2\mu_\B+\mu_\O} \frac{\mAB (\mu_\A+\mu_\B)}{1-\mAB} \\
 &= \frac{\mAB (\mu_\A+\mu_\B)}{2\mu_\A+2\mu_\B+\mu_\O} \\
 \\
 \beta &< \frac{\mAB (\mu_\A+\mu_\B)}{2\mu_\A+2\mu_\B+\mu_\O}
\end{align*}

The right hand side is maximal when $\mAB$ is maximal, and $\mAB,\mu_\A,\mu_\B,\mu_\O \rightarrow 1/4$. This gives the final bound on $\beta$,

\begin{align*}
 \beta &< \frac{(1/4) (1/2)}{5/4} = 1/10\\
 \beta &< 1/10
\end{align*}
\end{proof}

Combining Lemmas \ref{lem:maxbeta-small}, \ref{lem:maxbeta-mid11}, \ref{lem:maxbeta-mid13}, and \ref{lem:maxbeta-mid3}, we find that the price of fairness is zero when $\beta>1/8$.

\begin{reptheorem}{thm:maxbeta}
The price of fairness is zero when $\beta>1/8$.
\end{reptheorem}

\begin{proof}
There are only four matchings with nonzero price of fairness and $\beta>0$, which are described in Propositions \ref{prop:smallbeta}, \ref{prop:midbeta11}, \ref{prop:midbeta13}, and \ref{prop:midbeta3}. Lemmas \ref{lem:maxbeta-small}, \ref{lem:maxbeta-mid11}, \ref{lem:maxbeta-mid13}, and \ref{lem:maxbeta-mid3} state that the maximum $\beta$ for any of these matchings is $1/8$. When $\beta>1/8$, the matching is not one of these four cases, and the price of fairness is zero.
\end{proof}

Theorems \ref{thm:pofdecreases} and \ref{thm:maxbeta} are the two main theoretical results of this paper: adding NDDs to the random graph model does not increase the upper bound on the price of fairness, and when the proportion of NDDs is high enough ($\beta>1/8$), the price of fairness is zero. We show this by addressing each of the four efficient matchings on the random graph model with nonzero price of fairness. In each case, and $\beta<1/8$, and the matching with NDDs has a smaller price of fairness than the matching without NDDs given in~\englishcite{Dickerson14:Price}.

To further explore these results, we numerically find the maximum price of fairness for the matchings given in Propositions \ref{prop:smallbeta}, \ref{prop:midbeta11}, \ref{prop:midbeta13}, and \ref{prop:midbeta3}. For each matching, we find the maximum price of fairness for a range of $\beta$, within the defined constraints,  using the ``NMaximize'' function in Mathematica with the nonlinear interior point method. 

\begin{figure}[h]\label{fig:allcases}
\centering
\includegraphics[width=0.5\textwidth]{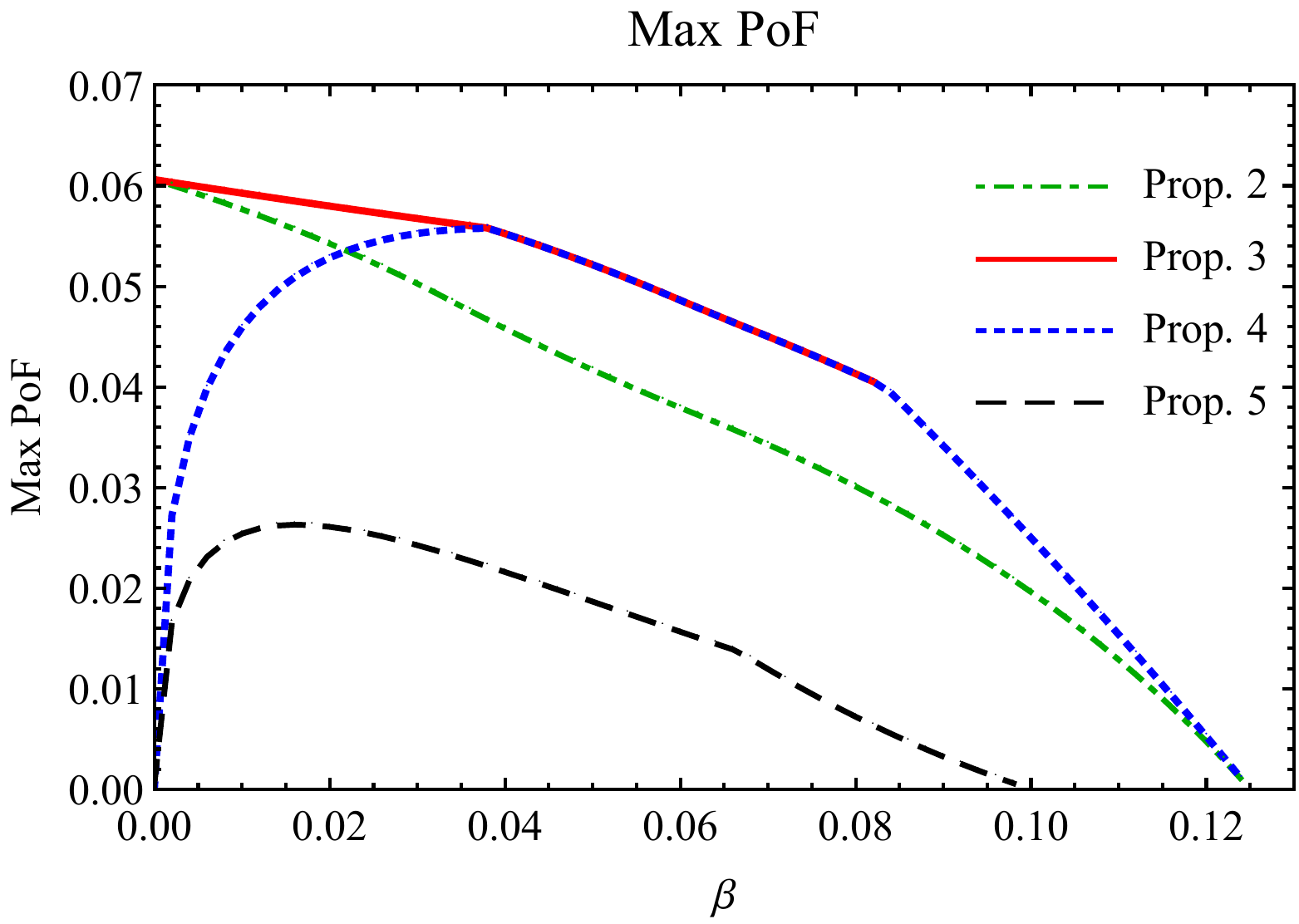}
\caption{Maixmum price of fairness for each of the four matchings addressed in Propositions \ref{prop:smallbeta}, \ref{prop:midbeta11}, \ref{prop:midbeta13}, and \ref{prop:midbeta3}. }
\end{figure}

Figure~\ref{fig:allcases} confirms both of our main theoretical results: adding NDDs to the efficient matching decreases the upper bound on the price of fairness, and when $\beta>1/8$ the price of fairness is zero.

\section{Price of Fairness for $\alpha$-Lexicographic-, Weighted-, and Hybrid-Lexicographic Fairness}\label{sec:poffairness}

This section presents Theorems and Proofs regarding the price of fairness for the lexicographic, weighted, and hybrid-lexicographic fairness rules.

\subsection{Lexicographic Fairness}
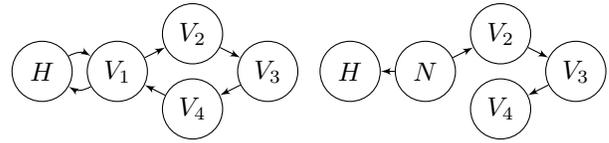
\begin{figure}
\centering
\begin{tikzpicture}
\tikzset{vertex/.style = {shape=circle,draw,minimum size=0.8cm}}
\tikzset{edge/.style = {->,> = latex'}}
\node[vertex] (h1) at  (0,0) {$H$};
\node[vertex] (v1) at  (1,0) {$V_1$};
\node[vertex] (v2) at  (2,0.5) {$V_2$};
\node[vertex] (v3)  at  (3,0) {$V_3$};
\node[vertex] (v4) at (2,-0.5) {$V_4$};
\draw[edge] (h1) to[bend left] (v1);
\draw[edge] (v1) to[bend left] (h1);
\draw[edge] (v1) to (v2);
\draw[edge] (v2) to (v3);
\draw[edge] (v3) to (v4);
\draw[edge] (v4) to (v1);
\end{tikzpicture} \hspace{3pt}
\begin{tikzpicture}
\tikzset{vertex/.style = {shape=circle,draw,minimum size=0.8cm}}
\tikzset{edge/.style = {->,> = latex'}}
\node[vertex] (h1) at  (0,0) {$H$};
\node[vertex] (a1) at  (1,0) {$N$};
\node[vertex] (v2) at  (2,0.5) {$V_2$};
\node[vertex] (v3)  at  (3,0) {$V_3$};
\node[vertex] (v4) at (2,-0.5) {$V_4$};
\draw[edge] (a1) to (h1);
\draw[edge] (v1) to (v2);
\draw[edge] (v2) to (v3);
\draw[edge] (v3) to (v4);
\end{tikzpicture}
\caption{Supporting graphs for Theorems~\ref{thm:lex-cycle} (left) and~\ref{thm:lex-chain} (right), with cycle cap 4 and chain cap 3, respectively.}\label{fig:lex-example}
\end{figure}

\begin{reptheorem}{thm:lex-cycle}
For any cycle cap $L$ there exists a graph $G$ such that the price of fairness of $G$ under the $\alpha$-lexicographic fairness rule with $0<\alpha\leq 1$ is bounded by $\POF{\Mcal{}}{u_{\alpha}}\geq \frac{L-2}{L}.$
\end{reptheorem}

\begin{proof}
Consider a kidney exchange graph consisting of one highly-sensitized patient $H$ and $L$ non-highly-sensitized patients $V_i$ that form a directed cycle of length $L$. A 2-cycle connects $H$ with one $V_i$, as shown in Figure~\ref{fig:lex-example}.  With a cycle cap of $L$, the optimal utilitarian matching has utility $L$, while the optimal lexicographic matching has utility $u_{\alpha}=2$, for any $0<\alpha\leq 1$. The price of fairness in this graph is $\POF{\mathcal{M}}{u_{\alpha}}=(L-2)/L$.
\end{proof}

\begin{reptheorem}{thm:lex-chain}
For any chain cap $R$ there exists a graph $G$ such that the price of fairness of $G$ under the $\alpha$-lexicographic fairness rule with $0<\alpha\leq 1$ is bounded by $\POF{\mathcal{M}}{u_{\alpha}} \geq \frac{R-1}{R}.$
\end{reptheorem}

\begin{proof}
Consider the graph used in the proof of Theorem \ref{thm:lex-cycle}, with vertex $V_2$ as an NDD rather than a pair. With a chain cap of $R$, the optimal utilitarian matching has utility $R$, while the optimal $\alpha$-lexicographic matching has utility $u_{\alpha}=1$ for any $0<\alpha\leq 1$. The price of fairness in this graph is $\POF{\mathcal{M}}{u_{\alpha}}=(R-1)/R$.
\end{proof}

\subsection{Weighted Fairness}

\begin{reptheorem}{thm:weight-cycle}
For any cycle cap $L$ and $\gamma \geq L-1$, there exists a graph $G$ such that the price of fairness of $G$ under the weighted fairness rule is bounded by $\POF{\mathcal{M}}{u_{WF}} \geq \frac{L-2}{L}.$
\end{reptheorem}

\begin{proof} 
Consider the graph used in the proof of Theorem \ref{thm:lex-cycle}, with all edge weights equal to 1. Weighted fairness increases the weight of the edge ending in $H$ to $(1+\gamma)$. The weighted utility of the 2-cycle is $2+\gamma$, while the weighted utility of the $L$-cycle is $L$. If $\gamma$ is chosen such that $\gamma \geq L-2$, then the 2-cycle will be chosen over the $L$-cycle, resulting in the price of fairness $\POF{\mathcal{M}}{u_{WF}}=(L-2)/L$.
\end{proof}

\begin{reptheorem}{thm:weight-chain}
For any chain cap $R$ and $\gamma \geq R-1$, there exists a graph $G$ such that the price of fairness of $G$ under the weighted fairness rule is bounded by $\POF{\mathcal{M}}{u_{WF}} \geq \frac{R-1}{R}.$
\end{reptheorem}

\begin{proof}
Consider the graph used in the proof of Theorem \ref{thm:lex-chain}, with all weights equal to 1. The weighted utility of the 1-chain is $1+\gamma$, while the weight of the $R$-chain is $R$. If $\gamma$ is chosen such that $\gamma \geq R-1$, then the 1-chain will be chosen over the $R$-chain, resulting in the price of fairness $\POF{\mathcal{M}}{u_{WF}}=(R-1)/R$.
\end{proof}

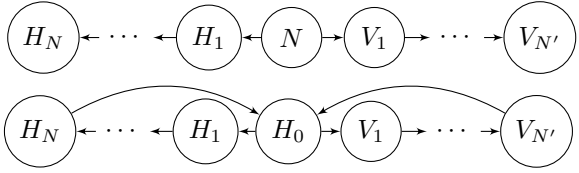
\begin{figure}
\centering
\begin{tikzpicture}
\tikzset{vertex/.style = {shape=circle,draw,minimum size=0.8cm}}
\tikzset{edge/.style = {->,> = latex'}}
\node[vertex] (A) at  (0,0) {$N$};
\node[vertex] (h1) at  (-1.1,0) {$H_1$};
\node[draw=none] (el) at  (-2.2,0) {$\cdots$};
\node[vertex] (hn)  at  (-3.3,0) {$H_N$};
\node[vertex] (v1) at (1.1,0) {$V_1$};
\node[draw=none] (el2) at (2.2,0) {$\cdots$};
\node[vertex] (vn)  at  (3.3,0) {$V_{N'}$};
\draw[edge] (A) to (h1);
\draw[edge] (h1) to (el);
\draw[edge] (el) to (hn);
\draw[edge] (A) to (v1);
\draw[edge] (v1) to (el2);
\draw[edge] (el2) to (vn);
\end{tikzpicture} 

\begin{tikzpicture}
\tikzset{vertex/.style = {shape=circle,draw,minimum size=0.8cm}}
\tikzset{edge/.style = {->,> = latex'}}
\node[vertex] (A) at  (0,0) {$H_0$};
\node[vertex] (h1) at  (-1.1,0) {$H_1$};
\node[draw=none] (el) at  (-2.2,0) {$\cdots$};
\node[vertex] (hn)  at  (-3.3,0) {$H_N$};
\node[vertex] (v1) at (1.1,0) {$V_1$};
\node[draw=none] (el2) at (2.2,0) {$\cdots$};
\node[vertex] (vn)  at  (3.3,0) {$V_{N'}$};
\draw[edge] (A) to (h1);
\draw[edge] (h1) to (el);
\draw[edge] (el) to (hn);
\draw[edge] (hn) to [bend left] (A);

\draw[edge] (A) to (v1);
\draw[edge] (v1) to (el2);
\draw[edge] (el2) to (vn);
\draw[edge] (vn) to [bend right] (A);
\end{tikzpicture}
\caption{Graphs for Theorems~\ref{thm:beta} (top) and~\ref{thm:betacycle} (bottom).}\label{betagraph}
\end{figure}

\begin{reptheorem}{thm:beta}
With no chain cap, there exists a graph $G$ such that the price of fairness of $G$ under the weighted fairness rule is bounded by $\POF{\mathcal{M}}{u_{WF}} \geq \frac{\gamma}{\gamma+1}.$
\end{reptheorem}

\begin{proof}
Consider a graph with a single NDD connected to a chain with highly-sensitized patients $H_i$ of length $N$, and a chain with non-highly sensitized patients $V_i$ of length $N'=\floor*{(\gamma+1) N}-1$. Under weighted fairness, the $V_i$ chain receives utility $u_L=\floor*{(\gamma+1)N}-1$ while the $H_i$ chain receives utility $u_H = (\gamma+1)N$, so $u_H>u_L$. The price of fairness for this graph is

{\small
$$\POF{\mathcal{M}}{u_{WF}} = \frac{\floor*{\gamma N_H}-1}{\floor*{\gamma N}+N-1} \geq \frac{\gamma N-2}{(\gamma+1)N-1}.$$
Taking the limit as $N \rightarrow \infty$ yields 
$$\lim_{N \rightarrow \infty} \frac{\gamma N-2}{(\gamma+1)N-1} = \frac{\gamma}{\gamma+1},$$
}
which implies $\POF{\mathcal{M}}{u_{WF}} \geq \frac{\gamma}{\gamma+1}.$
\end{proof}

\begin{reptheorem}{thm:betacycle}
With no cycle cap there exists a graph $G$ such that the price of fairness of $G$ under the weighted fairness rule is bounded by $\POF{\mathcal{M}}{u_{WF}} \geq \frac{\gamma}{\gamma+1}.$
\end{reptheorem}

\begin{proof}
Consider the graph used in the proof of Theorem \ref{thm:beta}, where the NDD $N$ is instead a highly-sensitized pair $H_0$, and the end vertices of both chains both have edges ending in $H_0$. Under weighted fairness, the $V_i$ cycle receives utility $u_L=\floor*{(\gamma+1)N}$, while the $H_i$ chain receives utility $u_H=(\gamma+1)N+1$, so $u_H>u_L$. The price of fairness for this graph is 
$$\POF{\mathcal{M}}{u_{WF}} = \frac{\floor*{\gamma}N-1}{\floor*{\gamma N}+N} \geq \frac{\gamma N - 1}{(\gamma + 1)N + 1}.$$

Taking the limit as $N \rightarrow \infty$ yields 

$$\lim_{N \rightarrow \infty} \frac{\gamma N - 1}{(\gamma + 1)N + 1} = \frac{\gamma}{\gamma+1},$$

\end{proof}

\subsection{Hybrid-Lexicographic}

\begin{reptheorem}{thm:hybridpof}
Assume the optimal utilitarian outcome $X_E$ receives utility $u(X_E)=u_E$, with one disadvantaged class that receives utility $u_1$, and $Z$ non-disadvantaged classes such that $u_1(X_E)>u_i(X_E)$. For $|\Pcal|$ classes, $ \POF{\mathcal{M}}{u_{\Delta}} \leq \frac{2((|\Pcal|-1)-Z)\Delta}{u_E}$.
\end{reptheorem}

\begin{proof}
Consider two outcomes, one in the fair regime ($X_F$), one in the utilitarian regime ($X_E$). Let $u_{\Delta }(X_F)> u_{\Delta }(X_E)$, such that $X_E$ receives nearly the same utility as $X_F$; that is, $u_{\Delta }(X_E) = u_{\Delta}(X_F)-\epsilon$ for some $0<\epsilon\ll 1$. WLOG, let there be $Z$ classes $i$ such that $u_1(X_E)>u_i(X_E)$, and 

{\small
\begin{equation} 
\begin{split}
u_{\Delta }(X_E) &= u_{\Delta }(X_F) - \epsilon \\
 &\leq \sum_{i=1}^{|\Pcal|} u_i(X_F) + (|\Pcal|-1)\Delta-\epsilon \\
 \end{split}
\end{equation}
}

Using the definition of utilitarian utility $u_E = \sum_{i=1}^{|\Pcal|} u_i$, 
\[
 u_E(X_E) - u_E(X_F)  \leq (2(|\Pcal|-1) - 2Z) \Delta-\epsilon
\]

and the price of fairness is 

$$\POF{\mathcal{M}}{u_{\Delta}} \leq \frac{2( (\Pcal-1) - Z) \Delta}{u_E(X_E)}.$$
\end{proof}

\section{Experimental Results}\label{sec:appendixresults}

This section contains worst-case price of fairness (PoF) and worst-case fairness ($\% F$) for real UNOS graphs, and for simulated graphs; these results were produced using the method described in \Secref{sec:experiments}.

\subsection{UNOS Graphs}

Figure \ref{fig:pofresults-unos} shows the worst-case (maximum) PoF of each fairness rule on the 314 UNOS graphs; Figure \ref{fig:fairresults-unos} shows worst-case (minimum) $\% F$.

\begin{sidewaysfigure*}
\centering
\includegraphics[width=\textwidth]{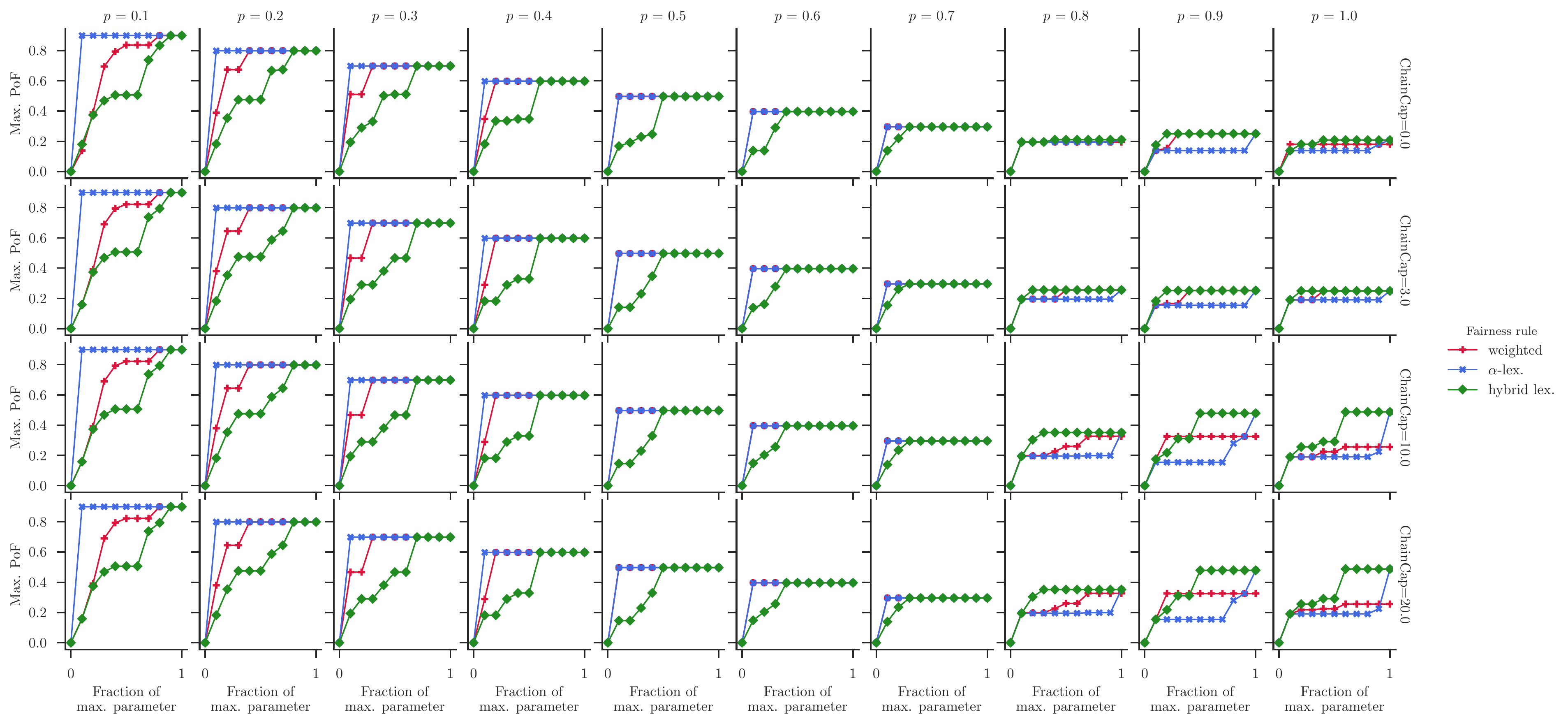} 
\label{fig:pofresults-unos}
\caption{Maximum PoF for each fairness rule. Parameters for each rule are $\alpha\in [0,1]$, $\beta\in [0,20]$, and $\Delta\in [0,u(M_E)]$. Rows correspond to edge success probabilities from $0.1$ to $1.0$; columns correspond to different chain caps: $0$, $3$, and $20$.}
\end{sidewaysfigure*}

Real exchange graphs are relatively sparse, and have very few feasible matchings. Each fairness rule effectively chooses one of these matchings, based on a fairness criteria. Especially with sparse graphs, fairness is often achieved by using longer cycles or cycles to match highly sensitized vertices. When edge success probability $p$ is high, fairness has little effect on overall utility, and PoF is often below 0.3. With lower edge success probability, using longer cycles and chains causes a huge loss in efficiency: the expected utility of $n$-cycles and chains is proportional to $p^n$, which incurs a huge penalty for long cycles and chains when $p$ is small. Thus as $p$ decreases, very small $\alpha$ and $\beta$ values result in a high PoF. Our results show that for  $p\leq 0.8$, even the smallest parameters for $\alpha$-lexicographic and weighted fairness ($\alpha=0.1$ and $\beta=2$) achieve the worst-case PoF. As expected, hybrid-lexicographic fairness limits PoF according to Theorem \ref{thm:hybridpof}. With two classes of patients (highly- and lowly-sensitized), the theoretical price of fairness is bounded by $\POF{\mathcal{M}}{u_\Delta}\leq 2\Delta /u(M_E)$; in the Figures, $\Delta$ is scaled by $u(M_E)$, so the upper bound on the price of fairness has a slope of two.

To illustrate the other side of the fairness-efficiency tradeoff, we consider worst case $\% F$. Figure \ref{fig:fairresults-unos} shows the minimum (worst case) $\% F$ over all UNOS graphs for each fairness rule, and for various edge success probabilities and chain caps.

\begin{sidewaysfigure*}
\centering
\includegraphics[width=\textwidth]{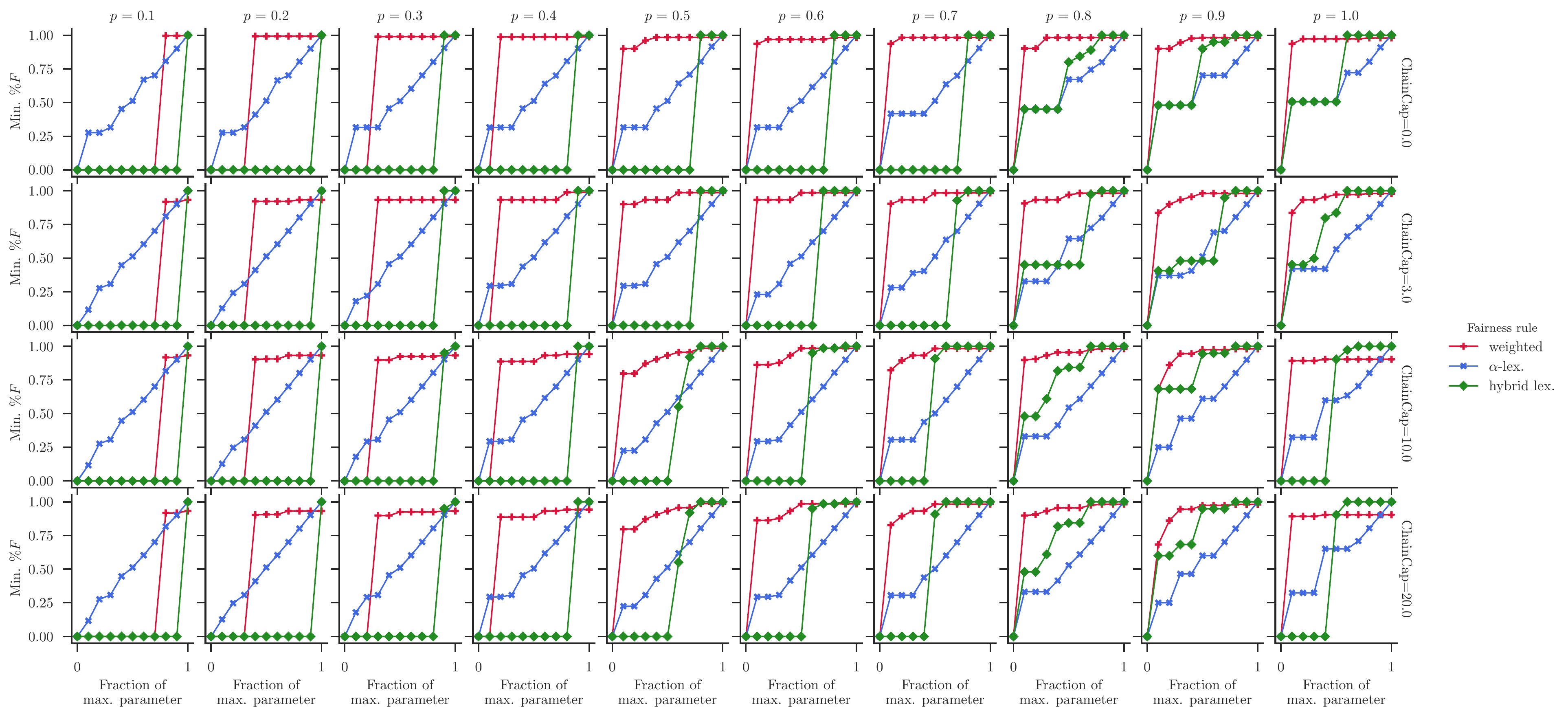}
\label{fig:fairresults-unos}
\caption{Minimum fraction of the fair score for each fairness rule. Parameters for each rule are $\alpha\in [0,1]$, $\beta\in [0,20]$, and $\Delta\in [0,u(M_E)]$. Rows correspond to edge success probabilities from $0.1$ to $1.0$; columns correspond to different chain caps: $0$, $3$, and $20$.}
\end{sidewaysfigure*}

As expected, $\alpha$-lexicographic fairness guarantees at $\% F \geq \alpha$;  weighted and hybrid-lexicogrpahic fairness do not make this guarantee. Small edge success probabilities make it impossible to match highly sensitized patients without large efficiency loss; when $p$ becomes small hybrid-lexicographic fairness matches no highly sensitized patients in the worst case.

These results demonstrate the balance between fairness and efficiency offered by both $\alpha$-lexicographic and hybrid-lexicographic fairness. If fairness is more important than efficiency, then the $\alpha$-lexicographic rule can be used to guarantee that the resulting matching achieves at least fraction $\alpha$ of the maximum possible fair score. Alternatively, if efficiency is more important than fairness, hybrid-lexicographic fairness can bound the price of fairness using parameter $\Delta$. 

\subsection{Simulated Exchange Graphs}

Simulated exchange graphs were drawn from previous UNOS exchanges, using the same method as \englishcite{Dickerson13:Failure}. These graphs are typically denser than real graphs, and have a much lower price of fairness. Figures \ref{fig:rand64} and \ref{fig:rand128} show the worst-case PoF and $\% F$ on 32 simulated exchanges of size 64 and 128. 

\begin{sidewaysfigure*}
\centering
\begin{subfigure}[Maximum PoF.]{\includegraphics[width=.49\linewidth] {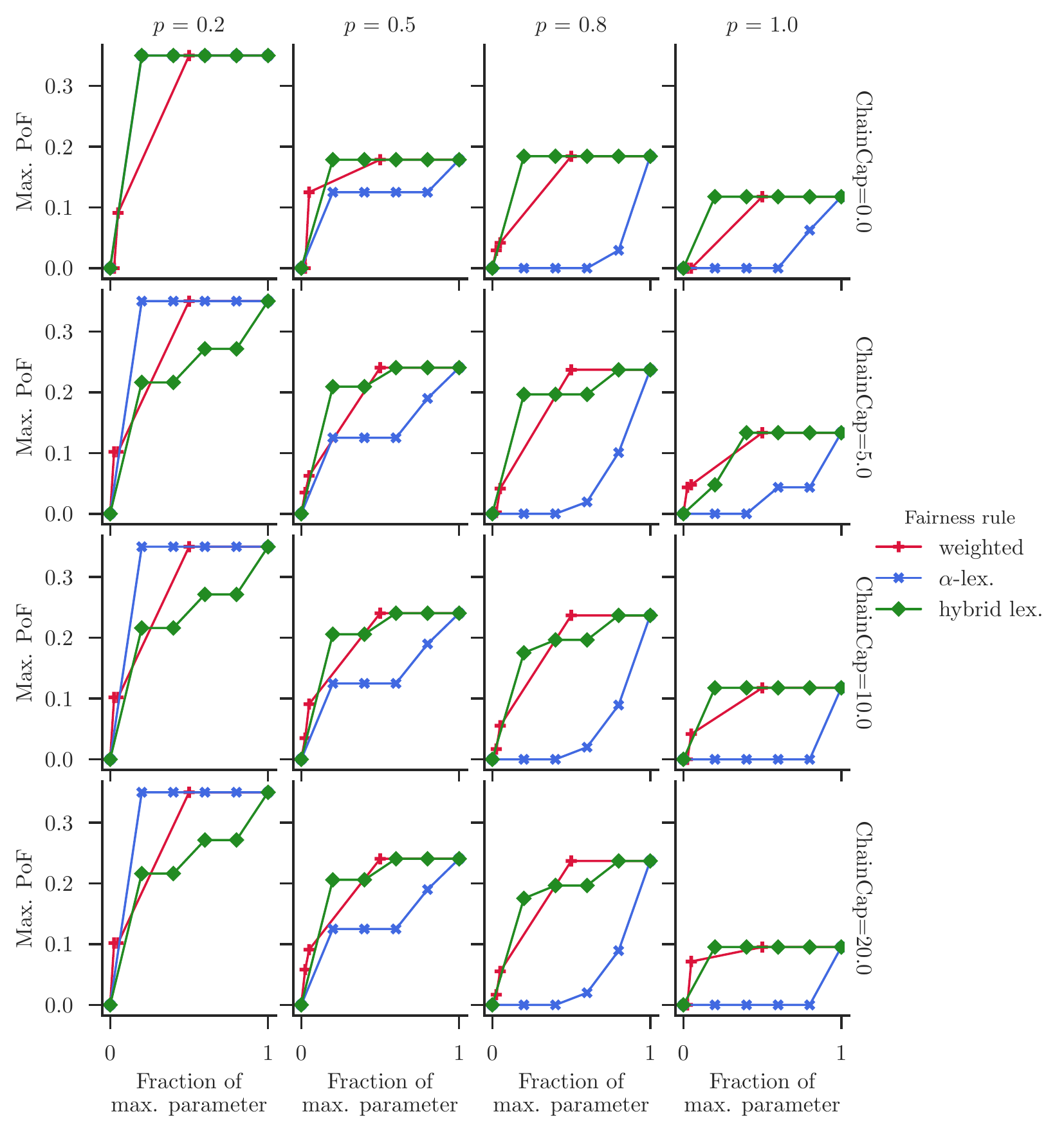}
   \label{fig:rand64pof}
 }%
 \end{subfigure}\hfill
 \begin{subfigure}[Minimum $\% F$.]{\includegraphics[width=.49\linewidth] {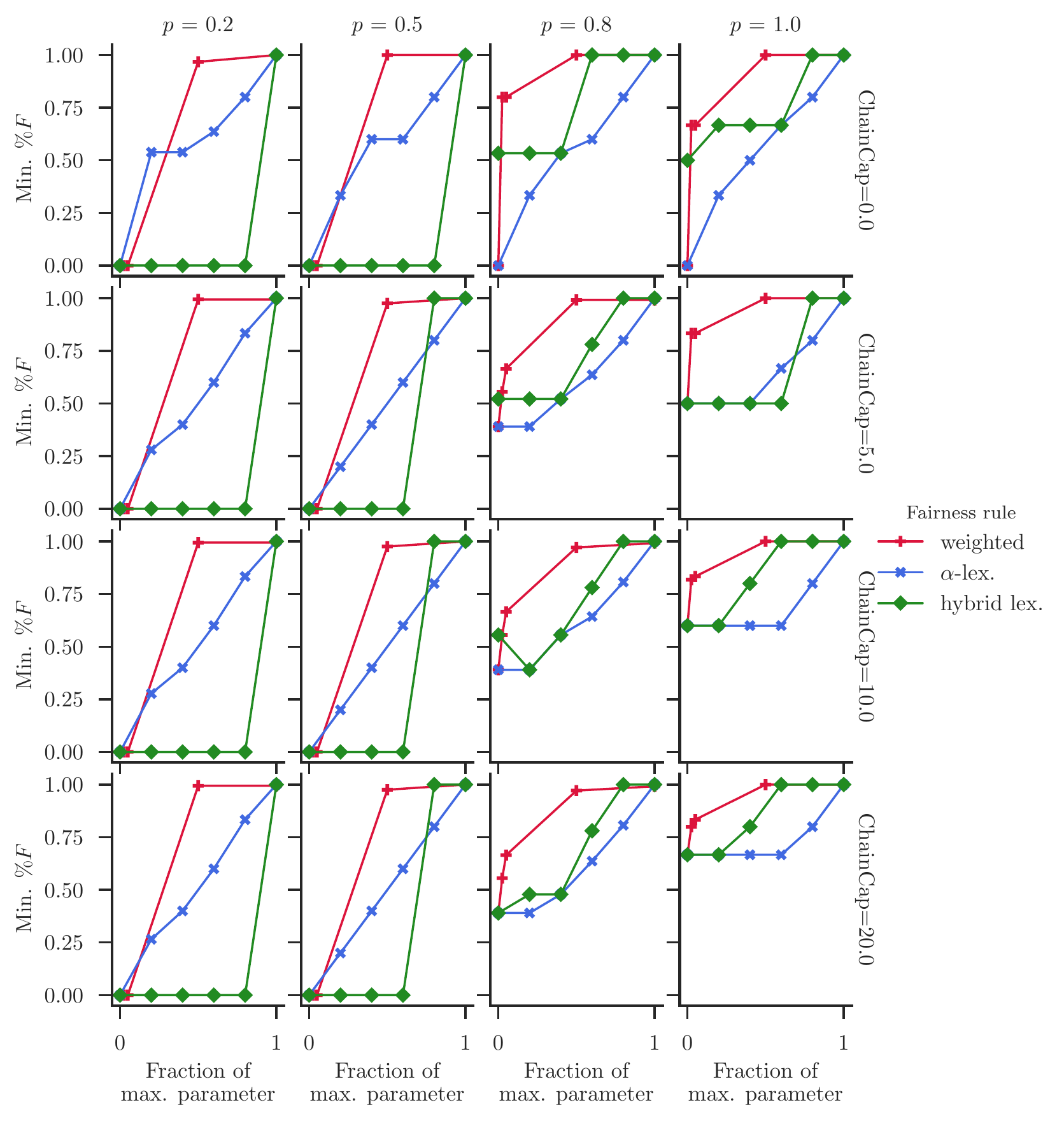}
   \label{fig:rand64fair}
 }%
\end{subfigure}
\caption{Worst-case PoF and $\% F$ for 32 64-vertex random graphs. Parameters for each rule are $\alpha\in [0,1]$, $\beta\in [0,20]$, and $\Delta\in [0,u(M_E)]$. Rows correspond to edge success probabilities from $0.1$ to $1.0$; columns correspond to different chain caps: $0$, $3$, and $20$.}
\label{fig:rand64}
\end{sidewaysfigure*}

\begin{sidewaysfigure*}
\centering
\begin{subfigure}[Maximum PoF.]{\includegraphics[width=.49\linewidth] {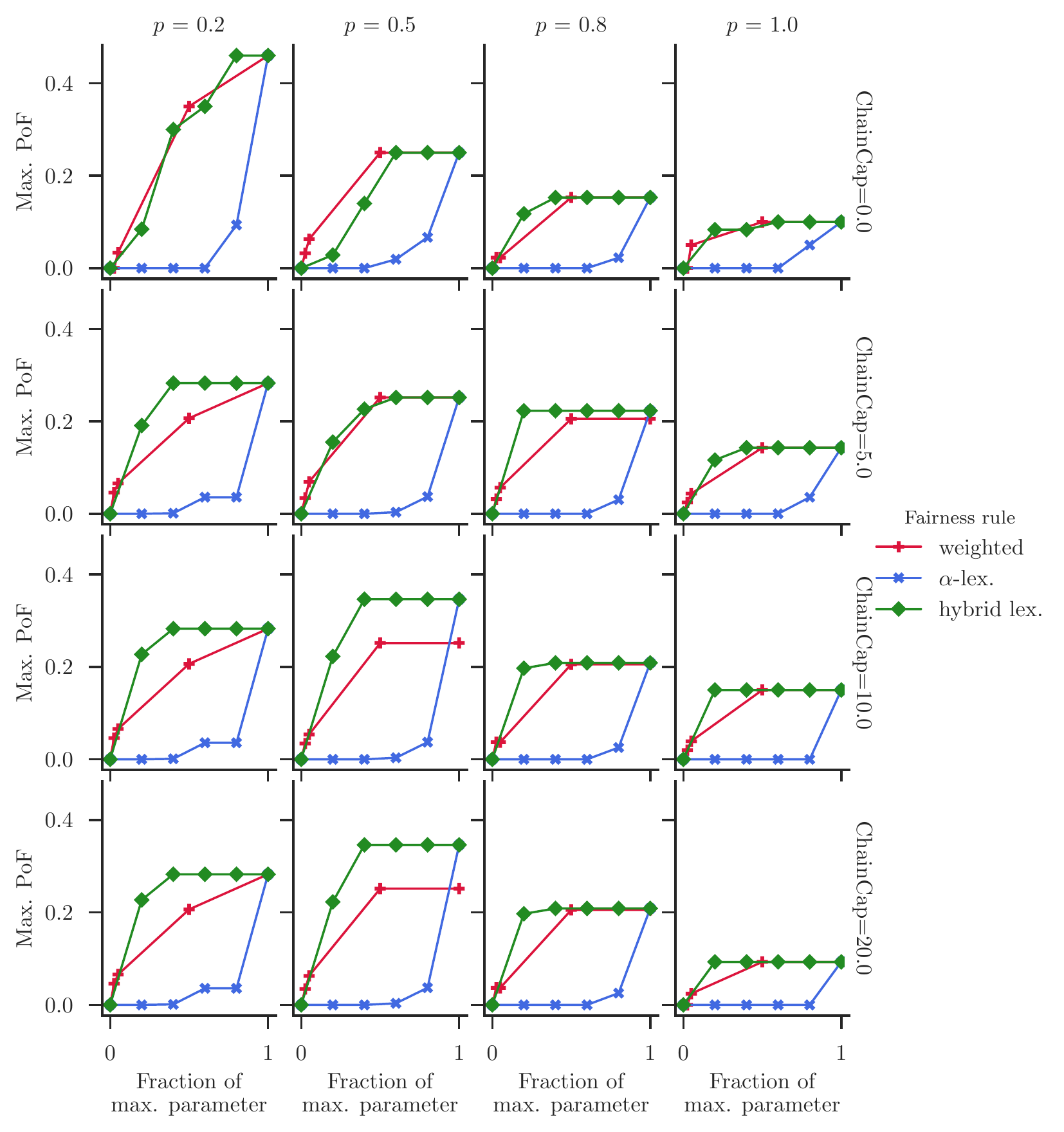}
   \label{fig:rand128pof}
 }%
 \end{subfigure}\hfill
 \begin{subfigure}[Minimum $\% F$.]{\includegraphics[width=.49\linewidth] {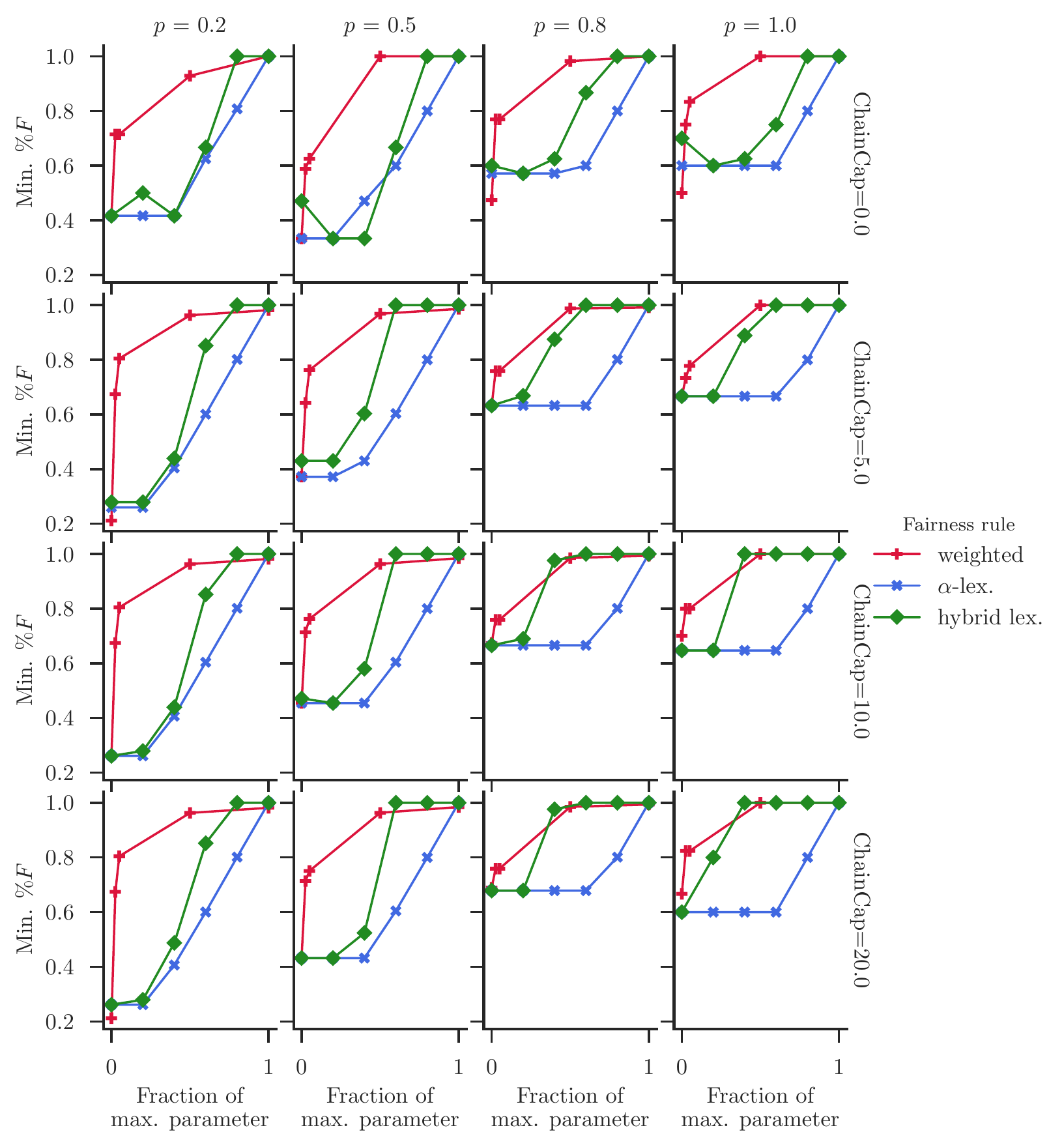}
   \label{fig:rand128fair}
 }%
\end{subfigure}
\caption{Worst-case PoF and $\% F$ for 32 128-vertex random graphs. Parameters for each rule are $\alpha\in [0,1]$, $\beta\in [0,20]$, and $\Delta\in [0,u(M_E)]$. Rows correspond to edge success probabilities from $0.1$ to $1.0$; columns correspond to different chain caps: $0$, $3$, and $20$.}
\label{fig:rand128}
\end{sidewaysfigure*}

\end{document}